\documentclass{siamltex1213}
\pdfoutput=1
\usepackage{array}
\usepackage{amsmath}
\usepackage{amssymb}
\usepackage{array}
\usepackage{float}
\usepackage{color}
\usepackage{graphicx}
\usepackage{epsfig}
\usepackage{graphics}
\usepackage{pstricks}

\newtheorem{remark}{Remark}

\newenvironment{proof1}{
    \noindent {\em Proof }}{\hfill$\Box$}

\begin{document}

\title{\bf Self-similar solutions for reversing interfaces in
the nonlinear diffusion equation with constant absorption}

\author{Jamie M. Foster$^{1}$ and Dmitry E. Pelinovsky$^{1,2}$ \\
    {\em $^1$ Department of Mathematics and Statistics, McMaster University, Hamilton ON, Canada, L8S 4K1}\\
    {\em $^2$ Department of Applied Mathematics, Nizhny Novgorod State Technical University,} \\
    {\em 24 Minin Street, Nizhny Novgorod, 603950, Russia}}

%\date{\today}
\maketitle

\begin{abstract}
We consider the slow nonlinear diffusion equation subject to a constant absorption rate and construct local self-similar solutions for reversing (and anti-reversing) interfaces, where an initially advancing (receding) interface gives way to a receding (advancing) one. We use an approach based on invariant manifolds, which allows us to determine the required asymptotic behaviour for small and large values of the concentration. We then `connect' the requisite asymptotic behaviours using a robust and accurate numerical scheme. By doing so, we are able to furnish a rich set of self-similar solutions for both reversing and anti-reversing interfaces.
\end{abstract}


\begin{AMS}\end{AMS}

\pagestyle{myheadings}
\thispagestyle{plain}
\markboth{J. M. FOSTER AND D. E. PELINOVSKY}{SELF-SIMILAR REVERSING INTERFACES}

\section{Introduction}

We address reversing and anti-reversing properties of interfaces in the following one-dimensional
slow diffusion equation with strong absorption
\begin{equation}
\label{heat}
\frac{\partial h}{\partial t} = \frac{\partial}{\partial x} \left( h^m \frac{\partial h}{\partial x} \right) - h^n,
\end{equation}
where $h$ is a positive function, \emph{e.g.}, a concentration of some species, and
$x$ and $t$ denote space and time, respectively. Restricting the exponents to the ranges $m > 0$ and $n < 1$ corresponds to the slow diffusion and strong absorption cases respectively.

Interfaces --- sometimes termed `contact lines' by fluid dynamicists ---
correspond to the points on the $x$-axis, where regions for positive solutions for $h$
are connected with the regions where $h$ is identically zero. The initial data $h|_{t=0} = h_0$ is assumed to be compactly supported. The motion of the interfaces is determined from conditions that require the function $h$ be continuous and the flux of $h$ through the interface to be zero \cite{Galaktionov}.

In the presence of slow diffusion ($m > 0$), the interfaces of compactly supported
solutions have a finite propagation speed \cite{Herrero}.
In the presence of strong absorption ($n < 1$), the solution vanishes for all $x$ after some finite time, which is referred to as finite-time extinction \cite{Kalash,Kersner}.
Therefore, the interfaces for a compactly supported initial data coalesce in a finite time. Depending on the shape of $h_0$ and the values of $m$ and $n$, the interfaces may change their direction of propagation in a number of different ways.
It was proved by Chen {\em et al.} \cite{Chen} that bell-shaped initial data
remains bell-shaped for all time before the compact support shrinks to a point. However, the possible types of
dynamics of interfaces for this bell-shaped data were not identified in \cite{Chen}.

The slow diffusion equation with the strong absorption (\ref{heat})
describes a variety of different physical processes, including: (i) the slow spreading of
a slender viscous film over a horizontal plate subject to the action of gravity and a constant
evaporation rate \cite{Acton} (when $m = 3$ and $n = 0$); (ii) the dispersion of a biological population subject to a constant
death-rate \cite{popn} (when $m = 2$ and $n = 0$); (iii) non-linear heat conduction along a rod with a constant
rate of heat loss \cite{Herrero} (when $m = 4$ and $n = 0$), and; (iv)
fluid flows in porous media with a drainage rate driven by gravity or background flows \cite{Aronson,PWH} (when $m = 1$ and
either $n = 1$ or $n = 0$).

Let us denote the location of the left interface by $x=\ell(t)$ and the limit
$x \searrow \ell(t)$, where $h$ is nonzero, by $x = \ell(t)^+$. If $m+n \geq 1$,
it was proved in \cite{Galaktionov} that the position of the interface, $\ell(t)$, is a Lipschitz continuous function of time $t$. In the case $m + n = 1$, the function $\ell(t)$ is found from the boundary conditions $h|_{x=\ell(t)} =0$
and
\begin{eqnarray}
\label{zero-flux}
\dot{\ell} = - h^{m-1} \frac{\partial h}{\partial x} \Big{|}_{x=\ell(t)^+} +  \left( h^{m-1} \frac{\partial h}{\partial x} \Big{|}_{x=\ell(t)^+} \right)^{-1},
\end{eqnarray}
where a dot denotes differentiation with respect to time. In the case $m + n > 1$,
the spatial derivatives at $x = \ell(t)^+$ are not well defined \cite{Galaktionov}.
and the zero flux condition (\ref{zero-flux}) must be rewritten as
\begin{eqnarray}
\dot{\ell} = \left\{ \begin{array}{l} \label{zero-flux2} \displaystyle - h^{m-1} \frac{\partial h}{\partial x} \Big{|}_{x=\ell(t)^+}, \quad
\mbox{\rm if} \;\; \dot{\ell} \leq 0, \\*[4mm]
\displaystyle h^{n} \left( \frac{\partial h}{\partial x} \right)^{-1} \Big{|}_{x=\ell(t)^+},
\quad \mbox{\rm if} \;\; \dot{\ell} \geq 0. \end{array} \right.
\end{eqnarray}

One could choose to close the slow diffusion equation (\ref{heat}) in a variety of ways, \emph{e.g.},
by supplying analogous conditions at the right interface, or by supplying a Dirichlet or Neumann condition elsewhere.
For instance, if $h_0$ is even in $x$, then the solution $h$ remains even in $x$ for all times,
and therefore, the slow diffusion equation (\ref{heat}) can be closed on the compact interval $[\ell(t),0]$ by imposing $\partial h/\partial x |_{x=0} = 0$.
However, such details do not concern us here because we are interested in studying the behaviour of solutions to
(\ref{heat}) local to the left interface $x = \ell(t)$ only.

We reiterate here the main question on the possible types of dynamics in the slow diffusion equation
with the strong absorption (\ref{heat}). Working with bell-shaped, compactly supported initial data $h_0$, one can anticipate
\emph{a priori} that the compact support of the bell-shaped solution can either: (i) decrease monotonically in time, or;
(ii) first expand and then subsequently shrink, or; (iii) have more complicated behaviour
where multiple instances of expansion and contraction are observed. This phenomenon brings
about both `reversing' and `anti-reversing' dynamics of an interface. Here the term `reversing' describes
a scenario where the velocity of the left interface $x = \ell(t)$ satisfies $\dot{\ell}<0$ before the reversing time
and $\dot{\ell}>0$ after the reversing time, whereas the term `anti-reversing' refers to the opposite scenario with
$\dot{\ell}>0$ before and $\dot{\ell}<0$ after the reversing time.

The first analytical solution to (\ref{heat}) exhibiting a reversing interface was obtained by Kersner \cite{Kersner} for the case $m + n = 1$.
This explicit solution takes the form
\begin{equation} \label{ker1}
u^{m}(x,t) = \frac{m}{2(m+1)(m+2) t} \left[ C t^{\frac{2}{m+2}} - (m+2)^2 t^2 - x^2 \right]_+,
\end{equation}
where the plus subscript denotes the positive part of the function, and $C > 0$ is an arbitrary parameter.
The interfaces are located symmetrically at $x = \pm \ell(t)$ with
\begin{equation} \label{ker2}
\ell(t) = \sqrt{C t^{\frac{2}{m+2}} - (m+2)^2 t^2}.
\end{equation}

More recently, Foster {\em et al.} \cite{Foster} considered the case $m + n > 1$
and explored the asymptotic and numerical construction of self-similar solutions
for equation (\ref{heat}) --- some related, yet different, self-similar solutions to other
nonlinear diffusion equations have previously been constructed using a combination of 
asymptotic analysis and numerical shooting; see, \emph{e.g.}, \cite{Foster2,Zhang}.

The self-similar solutions capture the relevant dynamics of reversing interfaces near the corresponding points in the space
and time (which can be placed at the origin of $x$ and $t$, without the loss of generality). Based on
a classical point symmetry analysis of the porous medium equation (\ref{heat}) --- provided in \cite{Gandarias} ---
the authors of \cite{Foster} found that the reversing interfaces can be described via the self-similar reductions
\begin{equation}
\label{reduction}
h(x,t) = \left( \pm t \right)^{\frac{1}{1-n}} \; H_{\pm}(\xi), \quad \xi = x (\pm t)^{-\frac{m+1-n}{2(1-n)}}, \quad \pm t > 0,
\end{equation}
where the functions $H_{\pm}$ satisfy a pair of second-order differential equations.

In this paper, we explore the case $n = 0$ only, when the functions $H_{\pm}$ of
the self-similar reduction (\ref{reduction}) satisfy the second-order differential equations
\begin{equation}
\label{ode}
\frac{d}{d \xi} \left( H_{\pm}^m \frac{d H_{\pm}}{d \xi} \right)
\pm \frac{m+1}{2} \, \xi \frac{d H_{\pm}}{d \xi} = 1 \pm H_{\pm}.
\end{equation}
We seek positive solutions $H_{\pm}$ of the differential equations (\ref{ode})
on the semi-infinite line $[A_{\pm},\infty)$ that satisfy the following conditions:
\begin{eqnarray}
\label{(i)}
\mbox{(i):} & & \qquad H_{\pm}(\xi) \to 0 \quad \mbox{\rm as} \quad \xi \to A_{\pm},\\
\label{(ii)} \mbox{(ii):} & & \qquad H_{\pm}(\xi) \;\; \mbox{\rm is monotonically increasing for all} \;\; \xi > A_{\pm}, \\
\label{(iii)}
\mbox{(iii):} & & \qquad H_{\pm}(\xi) \to +\infty \quad \mbox{\rm as} \quad \xi \to +\infty.
\end{eqnarray}
These first of these conditions, (\ref{(i)}), are the self-similar counterparts of the condition $h|_{x=\ell(t)}=0$ for the equation (\ref{heat}). In addition,
the behaviour of $H_-(\xi)$ and $H_+(\xi)$ in the far-field must be matched from the condition
\begin{equation} \label{far-field-matching}
\lim_{\xi \to \infty} \frac{H_+(\xi)}{H_-(\xi)} = 1.
\end{equation}
The requirement (\ref{far-field-matching}) is tantamount to enforcing that the solution
$h$ to the slow diffusion equation (\ref{heat}) does not `jump' as $t$ passes through zero ---
this can be verified by taking both the limits $t \searrow 0$ and $t \nearrow 0$ in
the self-similar reduction (\ref{reduction}).

Existence of solutions to the differential equations (\ref{ode}) on $[A_{\pm},\infty)$
with the required behaviour (\ref{far-field-matching}) implies, via the self-similar reduction (\ref{reduction}),
the existence of a reversing (if $A_\pm>0$) or anti-reversing (if $A_\pm<0$) left interface at $x = \ell(t)$, which behaves like
\begin{equation}
\label{coordinates}
\ell(t) = A_{\pm} (\pm t)^{\frac{m+1}{2}}, \quad \pm t > 0,
\end{equation}
after placing the reversing point at the origin of the space--time plane.
If $m > 1$, the velocity of the interface, $\dot{\ell}(t)$, changes sign continuously as $t$
passes through zero.

By combining formal asymptotic constructions of the solutions
$H_{\pm}$ near the small and large values with a numerical shooting method,
the authors of \cite{Foster} claimed that for integer values of $m = 2,3,4$,
there exists a unique positive value of $A_-$, which leads to a monotonically growing function
$H_-$ on the entire semi-axis $[A_-,\infty)$. Using the far-field matching condition (\ref{far-field-matching}),
a unique, monotonically growing function $H_+$ is found on $[A_+,\infty)$
for a positive value of $A_+$.
These solutions correspond to a reversing left interface $x = \ell(t)$ for $m = 2,3,4$.
No solutions exhibiting an anti-reversing interface were found in \cite{Foster}.

In the present work, we address the same problem using a dynamical system
framework \cite{ode-text,Wiggins}. The dynamical system theory allows us
to justify the formal asymptotic approximations of $H_{\pm}$ for small
and large values, as well as to set up an accurate and robust numerical
procedure for furnishing appropriate solutions to the differential equations (\ref{ode}).
Qualitatively, we recover the results of \cite{Foster} for $m = 3,4$, but with
a better precision, and we generalize these results for all non-integer values of $m>1$. In addition,
we demonstrate that the result for $m=2$, reported in \cite{Foster}, is incorrect and
no self-similar reversing interface solutions exist for $m = 2$. In addition,
we discover new reversing and anti-reversing interface solutions of
the same differential equations (\ref{ode}) for other values of $m$.

Our approach explores invariant manifolds for the singular
differential equations after appropriate unfolding (which is sometimes
referred to as the blow-up technique \cite{Dimortier,Sandstede}).
The main analytical results of this work are given by the following theorems.\\

\begin{theorem}
\label{theorem-center-manifold-zero}
For every $m > 1$ and $A_{\pm} \neq 0$, there exists a unique solution of the
differential equation (\ref{ode}) such that
$H_{\pm}(\xi) \to 0$ as $\xi \to A_{\pm}$. If $\pm A_{\pm} > 0$,
this unique solution has the following asymptotic behaviour
\begin{equation}
\label{center-manifold-zero-asymptotics}
H_{\pm}(\xi) = \pm \frac{2}{(m+1) A_{\pm}} (\xi - A_{\pm}) + \mathcal{O}((\xi - A_{\pm})^{\min\{2,m\}}), \quad
\mbox{\rm as} \quad \xi \to A_{\pm},
\end{equation}
whereas if $\pm A_{\pm} < 0$, it has the following asymptotic behaviour
\begin{equation}
\label{stable-manifold-zero-asymptotics}
H_{\pm}(\xi) = \left( \mp \frac{m (m+1) A_{\pm}}{2} (\xi - A_{\pm}) \right)^{\frac{1}{m}} +
\mathcal{O}(\xi - A_{\pm}), \quad \mbox{\rm as} \quad \xi \to A_{\pm}.
\end{equation}
\end{theorem}

\begin{theorem}
\label{theorem-center-manifold-infinity}
There exists a one-parameter family of solutions of the differential equation (\ref{ode}) for the lower sign such that
$H_-(\xi) \to +\infty$ as $\xi \to +\infty$, and this family has the following asymptotic behaviour
\begin{equation}
\label{center-manifold-infinity-asymptotics}
H_-(\xi) = \left( \frac{\xi}{x_0} \right)^{\frac{2}{m+1}} \left( 1 + \mathcal{O}(\xi^{-1}) \right), \quad
\mbox{\rm as} \quad \xi \to +\infty,
\end{equation}
where $x_0 > 0$ is an arbitrary parameter. There exists a two-parameter family of solutions
of the differential equation  (\ref{ode}) for the upper sign such that
$H_+(\xi) \to +\infty$ as $\xi \to +\infty$, and this family has the
same asymptotic behaviour (\ref{center-manifold-infinity-asymptotics}) for some $x_0 > 0$.
\end{theorem}

\vspace{0.25cm}

The main problem is to connect the two asymptotic behaviours of the differential equations
(\ref{ode}) which are defined for small and large values of $H_{\pm}$ by Theorems
\ref{theorem-center-manifold-zero} and \ref{theorem-center-manifold-infinity}.
We know from \cite{Foster} that there exists an exact solution of the connection problem
if $A_{\pm} = 0$. This exact solution is given by
\begin{equation}
\label{exact-solution}
H_{\pm}(\xi) = \left( \frac{m+1}{2} \xi^2 \right)^{\frac{1}{m+1}}, \quad \xi \in (0,\infty).
\end{equation}
However, the exact profile (\ref{exact-solution}) corresponds to a solution to
the slow diffusion equation (\ref{heat}) with an interface that remains stationary
for all time, thus we do not examine it further here.

Some connection results for nonzero values of $A_{\pm}$ are available for the differential equation (\ref{ode})
in Lemmas \ref{lemma-connection} and \ref{lemma-continuation} below. Although
these results are not sufficient for an
analytical solution of the connection problem, we can set up a numerical method,
which detects connections of the two solutions described
in Theorems \ref{theorem-center-manifold-zero} and \ref{theorem-center-manifold-infinity}.

The remainder of the paper is organized as follows. The unfolding and invariant manifolds
for the differential equations (\ref{ode}) near small values of $H_{\pm}$ are described in \S2.
The corresponding results near large values of $H_{\pm}$
are reported in \S3. The connection problem between
the invariant manifolds near small and large values
of $H_{\pm}$ is considered in \S4. The relevant numerical technique
is implemented in \S5, where the main findings are discussed and compared with
the previous results from \cite{Foster}. The paper is concluded in \S6
with a discussion of the relevance of the self-similar solutions to the dynamics of (\ref{heat}).

\section{Invariant manifolds for small values of $H_{\pm}$}

We shall rewrite the scalar equations (\ref{ode}) as vector systems
for variables $u = H_{\pm}$ and $w = H_{\pm}^m \frac{d H_{\pm}}{d \xi}$. In the interests of simplicity of notation, we drop the plus and minus subscripts in the definitions of the variables $u$ and $w$. The non-autonomous vector system for $u$ and $w$ is as follows:
\begin{equation}
\label{system}
\left\{ \begin{array}{l}
 \frac{du}{d\xi} = \frac{w}{u^m}, \\
 \frac{dw}{d\xi} = 1 \pm u \mp \frac{m+1}{2} \frac{\xi w}{u^m}. \end{array} \right.
\end{equation}
If $m$ is a non-integer, we require the constraint $u \geq 0$. In either case,
only positive solutions for $u$ are needed to be considered.

The interface, in self-similar variables, is assumed to be located at
$\xi = A \in \mathbb{R}$, where $u = 0$  --- a requirement of the condition on the continuity of $h(x,t)$ at $x=\ell(t)$.
Since the value $u = 0$ is singular in the non-autonomous system (\ref{system}),
we shall unfold the singularity by introducing a convenient parametrization
of solutions with a new time variable $\tau$ defined by
$$
\frac{d \xi}{d \tau} = u^m, \quad u \geq 0.
$$
The map $\tau \mapsto \xi$ is increasing and if $\xi \to A$ as $\tau \to -\infty$, then
$\xi > A$ for finite values of $\tau$.

With the parametrization $\tau \mapsto \xi$, we obtain
the autonomous dynamical system in $\mathbb{R}^3$,
\begin{equation}
\label{zero-dynamics}
\left\{ \begin{array}{l}
\dot{\xi} = u^m, \\
\dot{u} = w, \\
\dot{w} = u^m (1 \pm u) \mp \frac{m+1}{2} \xi w,\end{array} \right.
\end{equation}
where the dots stand for the derivatives of $(\xi,u,w)$ in $\tau$. In what follows,
we assume that $m \geq 1$, so that the vector field
of the dynamical system (\ref{zero-dynamics}) is continuously differentiable
near zero values of $u$. Again, $u \geq 0$ has to be enforced if $m$ is a non-integer.

The family of equilibrium points for the system (\ref{zero-dynamics})
is given by $(\xi,u,w) = (A,0,0)$, where $A \in \mathbb{R}$ is an arbitrary parameter.
If $m > 1$, each equilibrium point is associated with the Jacobian matrix
$$
\left[ \begin{array}{ccc} 0 & 0 & 0 \\ 0 & 0 & 1 \\ 0 & 0 & \mp \frac{m+1}{2} A \end{array} \right].
$$
This Jacobian matrix has a double zero eigenvalue (with two linearly independent eigenvectors)
and a simple nonzero eigenvalue $\mp \frac{m+1}{2} A$.
Therefore, the linearization of the dynamical system (\ref{zero-dynamics})
at the equilibrium point $(A,0,0)$ with $\pm A > 0$ has
a two-dimensional center manifold and a one-dimensional stable manifold, whereas
the linearized system with $\pm A < 0$ has a two-dimensional center manifold
and a one-dimensional unstable manifold. Since the dynamical system (\ref{zero-dynamics}) is $C^1$
smooth, a straightforward application of the invariant manifold theorems \cite{ode-text,Wiggins}
asserts that the equilibrium point $(A,0,0)$ with $A \neq 0$ is located at the intersection
of the two invariant manifolds, which are tangential to the invariant manifolds of the
linearized system. We formulate these results in the two following Propositions,
namely Propositions \ref{proposition-center-manifold-zero} and \ref{proposition-stable-manifold-zero}.
The relevant conclusions on the behaviour of solutions
of the differential equations (\ref{ode}) for small values of $H_{\pm}$,
expressed in Theorem \ref{theorem-center-manifold-zero}, follow from these two Propositions.\\

\begin{proposition}
\label{proposition-center-manifold-zero}
For every $m > 1$ and $A \neq 0$, there exists a two-dimensional center manifold
of the dynamical system (\ref{zero-dynamics}) near the equilibrium point
$(A,0,0)$, which can be parameterized as follows:
\begin{equation}
\label{center-manifold-zero}
W_c(A,0,0) = \left\{ w = \pm \frac{2 u^m}{(m+1) A} \left[ 1 + \mathcal{O}(\xi-A,u^{\min\{1,m-1\}}) \right],
\;\; u \in (0,\delta), \;\; \xi \in (A,A+\delta), \right\},
\end{equation}
where $\delta > 0$ is small. Dynamics of the system (\ref{zero-dynamics}) on
the center manifold $W_c(A,0,0)$ is topologically
equivalent to the dynamics at the truncated normal form
\begin{equation}
\label{center-manifold-zero-dynamics}
\left\{ \begin{array}{l}
\dot{\xi} = u^m, \\
\dot{u} = \pm \frac{2 u^m}{(m+1) A}.\end{array} \right.
\end{equation}
In particular, for every $A \neq 0$, there exists exactly one trajectory on $W_c(A,0,0)$,
which approaches the equilibrium point $(A,0,0)$ as $\tau \to -\infty$ if $\pm A > 0$ and
$\tau \to +\infty$ if $\pm A < 0$.
\end{proposition}

\vspace{0.25cm}

\begin{proof}
Existence of a two-dimensional center manifold $W_c(A,0,0)$, which is tangent to
that of the linearized system
$$
E_c(A,0,0) = \left\{ w = 0, \quad (\xi,u) \in \mathbb{R}^2 \right\},
$$
follows from Theorem 4.1 in \cite{Chiconi}. We develop an approximation
of $W_c(A,0,0)$ by writing
\begin{equation}\label{w-u-f}
w = u^m \eta(\xi,u),
\end{equation}
where $u^m \eta(\xi,u)$ is $C^1$ at
the point $(\xi,u) = (A,0)$ with zero partial derivatives.
Dynamics along $W_c(A,0,0)$ is given by the two-dimensional system
\begin{equation}
\label{zero-dynamics-center-dynamics}
\left\{ \begin{array}{l} \dot{\xi} = u^m, \\ \dot{u} = u^m \eta(\xi,u).\end{array} \right.
\end{equation}
The function $\eta$ is to be found by substituting (\ref{w-u-f})
to the three-dimensional system (\ref{zero-dynamics}) and using the two-dimensional
system (\ref{zero-dynamics-center-dynamics}). Then,
we obtain a partial differential equation
\begin{equation}
\label{zero-dynamics-center-manifold}
1 \mp \frac{m+1}{2} A \eta = \mp u \pm \frac{m+1}{2} (\xi - A) \eta + \eta \frac{\partial}{\partial u} (u^m \eta) + \frac{\partial}{\partial \xi} (u^m \eta).
\end{equation}
If $m > 1$ and $u^m \eta(\xi,u)$ is a $C^1$ function at $(\xi,u) = (A,0)$ with zero partial derivatives, then equation
(\ref{zero-dynamics-center-manifold}) has a solution such that
\begin{equation}
\label{approximation-center-manifold-zero}
\eta(\xi,u) = \pm \frac{2}{(m+1) A} + \mathcal{O}(\xi-A,u^{\min\{1,m-1\}}).
\end{equation}
The representation (\ref{w-u-f}) and (\ref{approximation-center-manifold-zero}) is equivalent to
(\ref{center-manifold-zero}). Substituting (\ref{approximation-center-manifold-zero})
to (\ref{zero-dynamics-center-dynamics}) and truncating
the remainder term, we obtain the truncated normal form (\ref{center-manifold-zero-dynamics}).

From the second equation of the system (\ref{center-manifold-zero-dynamics}), it follows that
if $\pm A > 0$, then $\dot{u} > 0$ such that
$u(\tau) \to 0$ as $\tau \to -\infty$, whereas if $\pm A < 0$, then $\dot{u} < 0$ such that $u(\tau) \to 0$ as $\tau \to +\infty$.
From the first equation of the system (\ref{center-manifold-zero-dynamics}), the constant of integration
for $\xi$ is arbitrary, so that $\xi(\tau) \to \tilde{A}$ in the same limit with $\tilde{A} \neq A$.
Hence, dynamics along the two-dimensional manifold $W_c(A,0,0)$ is decomposed between a curve of equilibrium
states with $u = 0$ and weakly unstable (if $\pm A > 0$) or weakly stable (if $\pm A < 0$) evolution
along a curve parameterized by small positive $u$.

Persistence of the dynamics on $W_c(A,0,0)$ with respect to the remainder terms in (\ref{approximation-center-manifold-zero})
follows from analysis of the system (\ref{zero-dynamics-center-dynamics}).
\end{proof}

\vspace{0.25cm}

\begin{proposition}
\label{proposition-stable-manifold-zero}
For every $m > 1$ and $\pm A < 0$, there exists a one-dimensional unstable manifold
of the dynamical system (\ref{zero-dynamics}) near the equilibrium point
$(A,0,0)$, which can be parameterized as follows:
\begin{equation}
\label{stable-manifold-zero}
W_{u}(A,0,0) = \left\{\xi = A + \mathcal{O}(u^m), \quad
w = \mp \frac{2 u}{(m+1)A} + \mathcal{O}(u^m), \quad u \in (0,\delta), \right\},
\end{equation}
where $\delta > 0$ is small. Dynamics of the system (\ref{zero-dynamics}) on
the unstable manifold $W_{u}(A,0,0)$ is topologically
equivalent to dynamics of the linear equation
\begin{equation}
\label{stable-manifold-zero-dynamics}
\dot{u} = \mp \frac{m+1}{2} A u.
\end{equation}
\end{proposition}

\begin{proof}
Existence of a one-dimensional unstable manifold $W_u(A,0,0)$, which is tangent to
that of the linearized system
$$
E_u(A,0,0) = \left\{ \xi = A, \;\; w = \mp \frac{m+1}{2} A u, \;\; u \in \mathbb{R} \right\},
$$
follows from Theorem 4.1 in \cite{Chiconi}.  We develop an approximation
of $W_u(A,0,0)$ by writing
\begin{equation}
\label{xi-w-in-terms-of-u}
\left\{ \begin{array}{l}
\xi = A + u^m \phi(u), \\
w = \mp \frac{m+1}{2} A u + u^m \theta(u), \end{array} \right.
\end{equation}
where $u^m \phi(u)$ and $u^m \theta(u)$ are $C^1$ with the zero derivative at
$u = 0$. Dynamics along $W_u(A,0,0)$ is given by the one-dimensional system
\begin{equation}
\label{zero-dynamics-stable-dynamics}
\dot{u} = \mp \frac{m+1}{2} A u + u^m \theta(u).
\end{equation}
The functions $\phi$ and $\theta$ are to be found by substituting (\ref{xi-w-in-terms-of-u})
to the three-dimensional system (\ref{zero-dynamics}) and using the one-dimensional
system (\ref{zero-dynamics-stable-dynamics}). Then,
we obtain a system of differential equations
\begin{equation}
\label{zero-dynamics-stable-manifold-1}
\left( m \phi(u) + u \frac{d\phi}{du}\right) \left( \mp \frac{m+1}{2} A + u^{m-1} \theta(u) \right) = 1
\end{equation}
and
\begin{equation}
\label{zero-dynamics-stable-manifold-2}
\left( \mp \frac{m+1}{2} A + u^{m-1} \theta(u) \right) \left( m \theta(u) + u \frac{d\theta}{du}\right) = 1 \pm u
+ \frac{m+1}{2} \phi(u) \left( \frac{m+1}{2} A u \mp u^m \theta(u) \right).
\end{equation}
If $m > 1$ while $u^m \phi(u)$ and $u^m \theta(u)$ are $C^1$ with the zero derivative at
$u = 0$, then system (\ref{zero-dynamics-stable-manifold-1}) and (\ref{zero-dynamics-stable-manifold-2})
has a solution such that
\begin{equation}
\label{approximation-stable-manifold-zero}
\phi(u) = \mp \frac{2}{m (m+1) A} + \mathcal{O}(u^{m-1}), \quad \theta(u) =  \mp \frac{2}{m (m+1) A} + \mathcal{O}(u^{\min\{1,m-1\}}).
\end{equation}
The representation (\ref{xi-w-in-terms-of-u}) and (\ref{approximation-stable-manifold-zero})
is equivalent to (\ref{stable-manifold-zero}). Substituting (\ref{approximation-stable-manifold-zero})
for $\theta(u)$ to (\ref{zero-dynamics-stable-dynamics}) and truncating the remainder terms at
$\mathcal{O}(u^m)$, we obtain the linear equation (\ref{stable-manifold-zero-dynamics}). Persistence of the linear dynamics
on $W_u(A,0,0)$ with respect to the remainder terms in $\theta(u)$ follows from analysis of the
differential equation (\ref{zero-dynamics-stable-dynamics}).
\end{proof}

\vspace{0.25cm}

\begin{remark}
For every $m > 1$ and $\pm A > 0$, one can construct a one-dimensional stable manifold
$W_s(A,0,0)$ of the dynamical system (\ref{zero-dynamics}) near the equilibrium point
$(A,0,0)$, which exists for $\xi > A$, $u > 0$, and $w < 0$.
However, this stable manifold does not contain trajectories that approach
the equilibrium point $(A,0,0)$ as $\tau \to -\infty$.
\end{remark}

\vspace{0.25cm}

\begin{proof1}{\em of Theorem \ref{theorem-center-manifold-zero}.}
For every $m > 1$ and $\pm A \neq 0$, Proposition \ref{proposition-center-manifold-zero}
states that the equilibrium state $(A,0,0)$ is connected by the trajectories
of the dynamical system (\ref{zero-dynamics}) with $u(\tau) > 0$ as $\tau \to -\infty$
if and only if $\pm A > 0$. In this case, there exists exactly one trajectory with $u > 0$
such that $u(\tau) \to 0$ as $\tau \to -\infty$. This trajectory belongs to
the center manifold $W_c(A,0,0)$, whose dynamics satisfy the system (\ref{zero-dynamics-center-dynamics}).
From this system, we obtain a first-order non-autonomous equation
\begin{equation}
\label{1st-non-auto}
\frac{du}{d \xi} = \eta(\xi,u) = \pm \frac{2}{(m+1) A} + \mathcal{O}(\xi-A,u^{\min\{1,m-1\}}), \quad
\mbox{\rm as} \quad \xi \to A, \quad u \to 0.
\end{equation}
Integrating (\ref{1st-non-auto})
near $\xi = A$, we recover the asymptotic behaviour (\ref{center-manifold-zero-asymptotics}).

For every $m > 1$ and $\pm A < 0$, Proposition \ref{proposition-stable-manifold-zero}
states that the equilibrium state $(A,0,0)$ is connected by exactly one trajectory
of the dynamical system (\ref{zero-dynamics}) with $u > 0$ and $u(\tau) \to 0$ as $\tau \to -\infty$.
This trajectory belongs to the unstable manifold $W_u(A,0,0)$ with the dynamics satisfying
equation (\ref{zero-dynamics-stable-dynamics}). From (\ref{xi-w-in-terms-of-u})
and (\ref{approximation-stable-manifold-zero}), we obtain
\begin{equation}
\label{2nd-non-auto}
\xi = A + u^m \left[ \mp \frac{2}{m (m+1) A} + \mathcal{O}(u^{m-1}) \right] \quad \mbox{\rm as} \quad u \to 0.
\end{equation}
Inverting this nonlinear equation near $\xi = A$, we recover the asymptotic behaviour
(\ref{stable-manifold-zero-asymptotics}).
\end{proof1}

\section{Invariant manifolds for large values of $H_{\pm}$}

The trajectories departing from the equilibrium point $(A,0,0)$ of the dynamical system
(\ref{zero-dynamics}) is expected to arrive at infinite values for $\xi$ and $u$.
In order to study the behaviour of trajectories near infinite values for $\xi$ and $u$,
we shall define $y = 1/u$, which maps an infinite value for $u$ to a zero value for $y$.
The other variables $\xi$ and $w$ must be adjusted accordingly for small values of $y$.
Let us consider the following scaling transformation,
\begin{equation}
\label{scaling}
\xi = \frac{x}{y^p}, \quad u = \frac{1}{y}, \quad w = \frac{z}{y^q},
\end{equation}
where $(x,y,w)$ is the set of new variables, and the positive parameters $p$ and $q$ are to be chosen
below. Substituting the transformation (\ref{scaling}) to the dynamical system (\ref{zero-dynamics}),
we obtain the following autonomous system in $\mathbb{R}^3$:
\begin{equation}
\label{scaling-dynamics}
\left\{ \begin{array}{l}
 \dot{x} = y^{p-m} - p x z y^{1-q}, \\
 \dot{y} = - z y^{2-q}, \\
 \dot{z} = y^{q-m-1} (\pm 1 + y) \mp \frac{m+1}{2} x z y^{-p} - q z^2 y^{1-q},\end{array} \right.
\end{equation}
where a dot still denotes a derivative with respect to the time variable $\tau$.

The system (\ref{scaling-dynamics}) is singular at $y = 0$, no matter what positive values of $p$ and $q$ are
chosen. To unfold the singularity, we can now introduce a convenient parametrization
of solutions with the time variable $s$ instead of the variable $\tau$ such that
$$
\frac{d \tau}{d s} = y^p, \quad y \geq 0.
$$
The map $s \mapsto \tau$ is increasing and we can consider solutions
parameterized by the new time variable $s$ such that $y \to 0$ as $s \to +\infty$
(which could correspond to a finite value for the old time variable $\tau$).

With the parametrization $s \mapsto \tau$, the system (\ref{scaling-dynamics})
can be rewritten in the equivalent form
\begin{equation}
\label{scaling-dynamics-equivalent}
\left\{ \begin{array}{l}
 x' = y^{2p-m} - p x z y^{p+1-q}, \\
 y' = - z y^{p+2-q}, \\
 z' = y^{p+q-m-1} (\pm 1 + y) \mp \frac{m+1}{2} x z - q z^2 y^{p+1-q},\end{array} \right.
\end{equation}
where a prime now denotes a derivative with respect to the new
time variable $s$. A suitable choice of parameters $p$ and $q$ is given by
\begin{equation}
\label{choice}
\left\{ \begin{array}{l} p+1 = q, \\ p+q = m+2, \end{array} \right. \quad \Rightarrow \quad
\left\{ \begin{array}{l} p = \frac{m+1}{2}, \\ q = \frac{m+3}{2}, \end{array} \right.
\end{equation}
so that the explicit form of the transformations (\ref{scaling}) is given by
\begin{equation}
\label{scaling-explicit}
\xi = \frac{x}{y^{\frac{m+1}{2}}}, \quad u = \frac{1}{y}, \quad w = \frac{z}{y^{\frac{m+3}{2}}}.
\end{equation}
The choice (\ref{choice}) ensures that the system (\ref{scaling-dynamics-equivalent}) is rewritten in the
simplest non-singular form with a quadratic vector field:
\begin{equation}
\label{infinity-dynamics}
\left\{ \begin{array}{l}
 x' = y - \frac{m+1}{2} x z, \\
 y' = - z y, \\
 z' = y (\pm 1 + y) \mp \frac{m+1}{2} x z - \frac{m+3}{2} z^2.\end{array} \right.
\end{equation}

The family of equilibrium points for the system (\ref{infinity-dynamics})
is given by $(x,y,z) = (x_0,0,0)$, where $x_0 \in \mathbb{R}$ is an arbitrary parameter.
Each equilibrium point is associated with the Jacobian matrix
$$
\left[ \begin{array}{ccc} 0 & 1 & -\frac{m+1}{2} x_0 \\ 0 & 0 & 0 \\ 0 & \pm 1 & \mp \frac{m+1}{2} x_0 \end{array} \right].
$$
This Jacobian matrix has a double zero eigenvalue (with two linearly independent eigenvectors)
and a simple eigenvalue $\mp \frac{m+1}{2} x_0$.

In the context of reversing and anti-reversing interfaces, we are interested
in the behaviour of solutions for which $\xi \to +\infty$ and $u \to +\infty$.
It follows from the transformation (\ref{scaling-explicit}), that this
requirement restricts our consideration to the family of critical points $(x_0,0,0)$ with $x_0 > 0$.
If $x_0 > 0$, the linearization of the dynamical system (\ref{infinity-dynamics})
at the equilibrium point $(x_0,0,0)$ has
a two-dimensional center manifold and a one-dimensional stable (upper sign) or unstable
(lower sign) manifold. Since the vector field of the dynamical system (\ref{infinity-dynamics}) is analytic
(quadratic), another straightforward application of the invariant manifold theorems \cite{ode-text,Wiggins}
yields that the equilibrium point $(x_0,0,0)$ with $x_0 > 0$ is located at the intersection
of the two invariant manifolds, which are tangential to the invariant manifolds of the
linearized system. We formulate these results in the next two Propositions, namely Propositions \ref{proposition-center-manifold-infinity} and \ref{proposition-stable-manifold-infinity}.
The relevant conclusions on the behaviour of solutions
of the differential equations (\ref{ode}) for large values of $H_{\pm}$
expressed in Theorem \ref{theorem-center-manifold-infinity} follow as a corollary
from these two Propositions.\\

\begin{proposition}
\label{proposition-center-manifold-infinity}
For every $x_0 > 0$, there exists a two-dimensional center manifold
of the dynamical system (\ref{infinity-dynamics}) near the equilibrium point
$(x_0,0,0)$, which can be parameterized as follows:
\begin{equation}
\label{center-manifold-infinity}
W_c(x_0,0,0) = \left\{ y = \frac{m+1}{2} \left[ x z \pm \left( 1 - \frac{m+1}{2} x_0^2 \right) z^2 + \mathcal{O}(3) \right],
\quad \begin{array}{l} x \in (x_0-\delta,x_0+\delta), \\ z \in (-\delta,\delta), \end{array} \right\},
\end{equation}
where $\delta > 0$ is small and $\mathcal{O}(3)$ denotes cubic terms in $x - x_0$ and $z$.
The dynamics of the system (\ref{infinity-dynamics}) on the center manifold $W_c(x_0,0,0)$ is topologically
equivalent to the dynamics at the truncated normal form
\begin{equation}
\label{center-manifold-infinity-dynamics}
\left\{ \begin{array}{l}
x' = \pm \frac{m+1}{2} \left( 1 - \frac{m+1}{2} x_0^2\right) z^2, \\
z' = -z^2.\end{array} \right.
\end{equation}
In particular, there exists exactly one trajectory on $W_c(x_0,0,0)$,
which approaches the equilibrium point $(x_0,0,0)$ as $s \to +\infty$.
\end{proposition}

\vspace{0.25cm}

\begin{proof}
Existence of a two-dimensional center manifold $W_c(x_0,0,0)$, which is tangent to
that of the linearized system
$$
E_c(x_0,0,0) = \left\{ y  = \frac{m+1}{2} x_0 z, \quad (x,z) \in \mathbb{R}^2 \right\},
$$
follows from Theorem 4.1 in \cite{Chiconi}. We develop an approximation
of $W_c(x_0,0,0)$ by writing $y = f(x,z)$ and expanding $f$ using a Taylor series in 
small values of both $x-x_0$ and $z$. Dynamics along $W_c(x_0,0,0)$ is given by the two-dimensional system
\begin{equation}
\label{infinity-dynamics-center-dynamics}
\left\{ \begin{array}{l} x' = f(x,z) - \frac{m+1}{2} xz, \\
z' = f(x,z) (\pm 1 + f(x,z)) \mp \frac{m+1}{2} xz - \frac{m+3}{2} z^2.\end{array} \right.
\end{equation}
The function $f$ is to be found from the partial differential equation
\begin{equation}
\label{infinity-dynamics-center-manifold}
\frac{\partial f}{\partial x} \left[ f - \frac{m+1}{2} xz \right] +
\frac{\partial f}{\partial z} \left[ f (\pm 1 + f) \mp \frac{m+1}{2} xz - \frac{m+3}{2} z^2 \right]
+ z f = 0.
\end{equation}
Equations (\ref{infinity-dynamics-center-dynamics}) and (\ref{infinity-dynamics-center-manifold})
suggest the following near-identity transformation of the function $f$ given by
\begin{equation}
\label{near-identity}
f(x,z) = \frac{m+1}{2} xz + z^2 g(x,z)
\end{equation}
After the near-identity transformation (\ref{near-identity}), the
dynamics along $W_c(x_0,0,0)$ are given by the two-dimensional system
\begin{equation}
\label{infinity-dynamics-center-dynamics-eq}
\left\{ \begin{array}{l} x' = z^2 g(x,z), \\
z' = z^2 \left[\pm g(x,z) - \frac{m+3}{2} + \left( \frac{m+1}{2} + z g(x,z) \right)^2 \right].\end{array} \right.
\end{equation}
The function $g$ is now to be found from the partial differential equation
\begin{equation*}
z g \left( \frac{m+3}{2} + z \frac{\partial g}{\partial x} \right)
+ \left( \frac{m+1}{2} x + 2 z g + z^2 \frac{\partial g}{\partial z} \right)
\left( \pm g + \left( \frac{m+1}{2} x + z g \right)^2 - \frac{m+3}{2} \right)
+ \frac{m+1}{2} x = 0,
\end{equation*}
which has a solution such that
\begin{equation}
\label{f-1}
g(x,z) =  \pm \frac{m+1}{2} \left( 1 - \frac{m+1}{2} x_0^2 \right) + \mathcal{O}(x-x_0,z).
\end{equation}
Expansions (\ref{near-identity}) and (\ref{f-1}) yield (\ref{center-manifold-infinity}). Substituting (\ref{f-1}) into
(\ref{infinity-dynamics-center-dynamics-eq}) and truncating at the cubic terms,
we obtain the truncated normal form (\ref{center-manifold-infinity-dynamics}).

From the second equation of the system (\ref{center-manifold-infinity-dynamics}), it follows that
there exists a unique solution such that $z(s) \to 0$ as $s \to +\infty$ with $z > 0$.
From the first equation of the system (\ref{center-manifold-infinity-dynamics}), the constant of integration
for $x$ is arbitrary, so that $x(s) \to \tilde{x}_0$ as $s \to +\infty$ with $\tilde{x}_0 \neq x_0$.
Hence, dynamics along the two-dimensional manifold $W_c(x_0,0,0)$ is decomposed between a curve of equilibrium
states with $z = 0$ and weakly stable evolution along a curve parameterized by small positive $z$.
Persistence of the dynamics on $W_c(x_0,0,0)$ with respect to the remainder terms in (\ref{f-1}) follows from
analysis of the system (\ref{infinity-dynamics-center-dynamics-eq}).
\end{proof}

\vspace{0.25cm}

\begin{proposition}
\label{proposition-stable-manifold-infinity}
For every $x_0 > 0$, there exists a one-dimensional stable (upper sign) or unstable
(lower sign) manifold of the dynamical system (\ref{infinity-dynamics}) near the equilibrium point
$(x_0,0,0)$, which can be expressed explicitly:
\begin{equation}
\label{stable-manifold-infinity}
W_{s/u}(x_0,0,0) = \left\{ y = 0, \quad z = \mp \frac{m+1}{2} x \left[1 - \left(\frac{x}{x_0}\right)^{\frac{2}{m+1}} \right],
\quad x \in (x_0-\delta,x_0+\delta), \right\},
\end{equation}
where $\delta > 0$ is small.
The dynamics of the system (\ref{infinity-dynamics}) on the manifold $W_{s/u}(x_0,0,0)$ is topologically
equivalent to the dynamics of the linear equation
\begin{equation}
\label{stable-manifold-infinity-dynamics}
x' = \mp \frac{m+1}{2} x_0 (x - x_0).
\end{equation}
\end{proposition}

\begin{proof}
The linearized system has the stable/unstable manifold:
$$
E_{s/u}(x_0,0,0) = \left\{ y = 0, \quad z = \pm (x-x_0), \quad x \in \mathbb{R}\right\}.
$$
Existence of a one-dimensional manifold $W_{s/u}(x_0,0,0)$ that is tangent to
$E_{s/u}(x_0,0,0)$ follows from Theorem 4.1 in \cite{Chiconi}.
We notice from (\ref{infinity-dynamics}) that $y = 0$ is an invariant reduction
of the three-dimensional system. Therefore, we can set $y = 0$ and seek
a parametrization of $W_{s/u}(x_0,0,0)$ by working with $z = \psi(x)$,
where $\psi(x_0) = 0$ and $\psi'(x_0) = \pm 1$. Dynamics along $W_{s/u}(x_0,0,0)$ are given by the one-dimensional system
\begin{equation}
\label{infinity-dynamics-stable-dynamics}
x' = - \frac{m+1}{2} x \psi(x).
\end{equation}
From the three-dimensional system (\ref{infinity-dynamics}), we obtain
a linear differential equation for nonzero $\psi$:
\begin{equation}
\label{infinity-dynamics-stable-manifold}
x \frac{d\psi}{dx} = \frac{m+3}{m+1} \psi(x) \pm x.
\end{equation}
This equation completed with the initial condition $\psi(x_0) = 0$ admits a unique solution
\begin{equation}
\label{approximation-stable-manifold-infinity}
\psi(x) = \mp \frac{m+1}{2} x \left[1 - \left(\frac{x}{x_0}\right)^{\frac{2}{m+1}} \right].
\end{equation}
Note that $\psi'(x_0) = \pm 1$, because $W_{s/u}(x_0,0,0)$ is tangent to $E_{s/u}(x_0,0,0)$.
Substituting 
\begin{equation}
\label{expansion-h}
\psi(x) = \pm (x-x_0) + \mathcal{O}((x-x_0)^2)
\end{equation}
into the differential equation (\ref{infinity-dynamics-stable-dynamics})
and truncating at the quadratic remainder term,
we obtain the linear equation (\ref{stable-manifold-infinity-dynamics}). Persistence of the linear dynamics
on $W_{s/u}(x_0,0,0)$ with respect to the remainder term follows from analysis of the
system (\ref{infinity-dynamics-stable-dynamics}) with the expansion (\ref{expansion-h}).
\end{proof}

\vspace{0.25cm}

\begin{proof1}{\em of Theorem \ref{theorem-center-manifold-infinity}.}
For every $x_0 > 0$, Proposition \ref{proposition-center-manifold-infinity}
states that there exists exactly one trajectory with $y > 0$ such that
$y(s) \to 0$ as $s \to +\infty$.  This trajectory belongs to
the center manifold $W_c(x_0,0,0)$, whose dynamics satisfy
the system (\ref{infinity-dynamics-center-dynamics}). 
We recover the asymptotic behaviour (\ref{center-manifold-infinity-asymptotics}) 
by eliminating $y$ from the transformation (\ref{scaling-explicit}) and setting $x = x_0$.

For every $x_0 > 0$, Proposition \ref{proposition-stable-manifold-infinity}
states that the point $(x_0,0,0)$ is an intersection of the
two-dimensional center manifold $W_c(x_0,0,0)$
and the one-dimensional stable/unstable manifold $W_{s/u}(x_0,0,0)$
for the upper/lower sign. Therefore, the point can be reached
in the direction $s \to +\infty$ along $W_s(x_0,0,0)$ but can not
be reached along $W_u(x_0,0,0)$. This guarantees that the
trajectory in Proposition \ref{proposition-center-manifold-infinity}
with $y > 0$ such that $y(s) \to 0$ as $s \to +\infty$ is unique for the lower
sign.  Therefore, for every $x_0 > 0$, there exists a one-dimensional
set of solutions of the differential equation (\ref{ode}) for the lower sign
such that $H_-(\xi) \to +\infty$ as $\xi \to +\infty$.

On the other hand, the span of the trajectory in Proposition
\ref{proposition-center-manifold-infinity} and the trajectory
in Proposition \ref{proposition-stable-manifold-infinity}
is a two-dimensional set that hosts all trajectories with $y > 0$ such that
$y(s) \to 0$ as $s \to +\infty$. Therefore, for every $x_0 > 0$, 
there exists a two-dimensional
set of solutions of the differential equation (\ref{ode}) for the upper sign
such that $H_+(\xi) \to +\infty$ as $\xi \to +\infty$.
The rate of change along $W_s(x_0,0,0)$ is exponential in $s$
and the rate of change along $W_c(x_0,0,0)$ is algebraic in $s$.
Therefore, solutions along the two-dimensional set still obey
the asymptotic behaviour (\ref{center-manifold-infinity-asymptotics}). 
\end{proof1}

\section{Connection of invariant manifolds}

Let us summarize the results of the previous two sections on the existence of
solutions $H_{\pm}$ to the differential equations (\ref{ode}) on the semi-infinite
line $(A_{\pm},\infty)$ that satisfy the properties (\ref{(i)})-(\ref{(iii)}). Such solutions are related to the trajectories of the dynamical systems (\ref{zero-dynamics})
and (\ref{infinity-dynamics}), which depart from the equilibrium points where $H_{\pm}$ is zero
and arrive to the equilibrium point where $H_{\pm}$ is infinite. In what follows,
we will consider separately the two different systems for $H_+$ and $H_-$.

\subsection{\label{H-plus-connect}The system for $H_+$ ($t > 0$)}

By Propositions \ref{proposition-center-manifold-zero} and \ref{proposition-stable-manifold-zero},
for every nonzero $A \equiv A_+$, there is a unique trajectory
of the dynamical system (\ref{zero-dynamics}) in variables $(\xi,u,w)$ that departs from
the equilibrium point $(A_+,0,0)$ as $\tau \to -\infty$ and belongs to the domain $u > 0$.
This trajectory is contained in the center manifold $W_c(A_+,0,0)$ if $A_+ > 0$ and
in the unstable manifold $W_u(A_+,0,0)$ if $A_+ < 0$.

By Propositions \ref{proposition-center-manifold-infinity} and \ref{proposition-stable-manifold-infinity},
for every $x_0 > 0$, there is a two-dimensional set of trajectories
of the dynamical system (\ref{infinity-dynamics}) in variables $(x,y,z)$
that reaches the equilibrium point $(x_0,0,0)$ as $s \to +\infty$ and belongs
to the domain $y > 0$. This trajectory is contained in the intersection
between the center $W_c(x_0,0,0)$ and stable $W_s(x_0,0,0)$ manifolds.

We shall now establish that the same trajectory departing from
the equilibrium point $(A_+,0,0)$ in (\ref{zero-dynamics})
arrives to the  equilibrium point $(x_0,0,0)$ in (\ref{infinity-dynamics}).
This trajectory determines a unique solution $H_+$ of the differential equation
(\ref{ode}) with the upper sign satisfying properties (\ref{(i)})--(\ref{(iii)}).\\

\begin{lemma}
\label{lemma-connection}
Fix $A_+ \in \mathbb{R} \backslash \{0\}$ and
consider a one-parameter trajectory of the dynamical system (\ref{zero-dynamics}) for
the upper sign such that $(\xi,u,w) \to (A_+,0,0)$ as $\tau \to -\infty$ and $u > 0$. Then, there exists
a $\tau_0 \in \mathbb{R}$ (or $\tau_0 = +\infty$)
such that $\xi(\tau) \to +\infty$ and $u(\tau) \to +\infty$ as $\tau \to \tau_0$.
\end{lemma}

\vspace{0.25cm}

\begin{proof}
We introduce the energy-like quantity for the dynamical system (\ref{zero-dynamics}) with
the upper sign:
\begin{equation}
\label{energy-like}
E(u,w) := \frac{1}{2} w^2 - \frac{1}{m+1} u^{m+1} - \frac{1}{m+2} u^{m+2}.
\end{equation}
Computing the derivative of $E$ in $\tau$ along a solution of system (\ref{zero-dynamics}), we obtain
\begin{equation}
\label{energy-like-derivative}
\frac{dE}{d\tau} = -\frac{m+1}{2} \xi w^2.
\end{equation}

If $A_+ > 0$, then $\xi(\tau) \geq A_+ > 0$, and $E$ is a strictly decreasing function of $\tau$ as
long as the solution to system (\ref{zero-dynamics}) exists and $w(\tau) \neq 0$.
By the representation (\ref{center-manifold-zero})
of $W_c(A_+,0,0)$ in Proposition \ref{proposition-center-manifold-zero},
we have $w(\tau) > 0$ for sufficiently large negative $\tau$. Now $w(\tau)$ cannot vanish
for any $\tau$, because if $w = 0$, then $\dot{w} = u^m (1+u) > 0$,
which contradicts positivity of $w$ before vanishing.
Therefore, $E(w,u)$ decreases to $-\infty$ in finite or infinite time $\tau$.
Because $E(0,u) \leq E(w,u)$, we have $u \to +\infty$ if $E(w,u) \to -\infty$.

If $A_+ < 0$, then $E$ is an increasing function of $\tau$ at least for sufficiently
large negative $\tau$. By the representation (\ref{stable-manifold-zero})
of $W_u(A_+,0,0)$ in Proposition \ref{proposition-stable-manifold-zero},
we still have $w(\tau) > 0$ for sufficiently large negative $\tau$.
Also, $w(\tau)$ cannot vanish for any $\tau$ by the same contradiction,
since if $w = 0$, then $\dot{w} = u^m (1+u) > 0$. Therefore, $\xi(\tau)$ and
$u(\tau)$ are still increasing functions, and there is a finite $\tau_0 \in \mathbb{R}$
such that $\xi(\tau_0) = 0$ and $\dot{\xi}(\tau_0) > 0$. For
$\tau > \tau_0$, the energy method described above proves again 
that $u(\tau) \to +\infty$ in finite or infinite time $\tau$.

We shall now prove that $\xi(\tau) \to +\infty$ as $\tau \to \tau_0$. Since $w > 0$
for all $\tau \in (-\infty,\tau_0)$, the map $\tau \to u$ is monotonically increasing,
so that we can parameterize both $\xi$ and $w$ by $u$ and consider the limit $u \to +\infty$.
From the last two equations of the system (\ref{zero-dynamics}), we obtain
\begin{equation}
\label{equation-w}
w \frac{dw}{d u} = u^m (1 + u ) - \frac{m+1}{2} \xi w.
\end{equation}
Since $\xi(\tau) > 0$ for $\tau$ close to $\tau_0$
and $w(\tau) > 0$ for all $\tau \in (-\infty,\tau_0)$, we estimate
$$
\frac{d}{du} \left( \frac{1}{2} w^2 \right) \leq u^m (1 + u).
$$
Integrating this differential inequality, we obtain
\begin{equation}
\label{bound-w}
w^2 \leq C + \frac{2}{m+1} u^{m+1} + \frac{2}{m+2} u^{m+2},
\end{equation}
where $C > 0$ is a constant of integration. Therefore, as $u \to \infty$,
there exists a constant $w_{\infty} > 0$ such that $w \leq w_{\infty} u^{\frac{m+2}{2}}$
for sufficiently large $u$. From the first two equations of the system (\ref{zero-dynamics}),
we obtain
\begin{equation}
\label{equation-xi}
\frac{d \xi}{d u} = \frac{u^m}{w} \geq  \frac{u^{\frac{m-2}{2}}}{w_{\infty}}.
\end{equation}
Integrating this differential inequality, we obtain a lower bound
for $\xi$ given by
\begin{equation}
\label{bound-xi}
\xi \geq C + \frac{2}{m w_{\infty}} u^{\frac{m}{2}}
\end{equation}
where $C > 0$ is another constant of integration. This lower bound
yields $\xi \to +\infty$ as $u \to +\infty$.
\end{proof}

\vspace{0.25cm}

\begin{corollary}
\label{coral-1}
The trajectory of Lemma \ref{lemma-connection} such that $\xi(\tau) \to +\infty$ and $u(\tau) \to +\infty$
corresponds to the trajectory of system (\ref{infinity-dynamics}) approaching the equilibrium state $(x_0,0,0)$
for some $x_0 \in [0,\infty)$. Consequently, one can define a piecewise $C^1$ map
\begin{equation}
\label{two-maps}
\mathbb{R} \backslash \{0\} \ni A_+ \mapsto x_0 \in \mathbb{R}^+.
\end{equation}
\end{corollary}

\vspace{0.25cm}

\begin{proof}
We can use the transformation (\ref{scaling-explicit}). Since $u \to \infty$, then $y \to 0$.
From the bound (\ref{bound-w}), we obtain
\begin{equation}
\label{bound-z}
z = w u^{-\frac{m+3}{2}} \leq w_{\infty} u^{-\frac{1}{2}}.
\end{equation}
Therefore, $z \to 0$ as $u \to \infty$. On the other hand, from the bound (\ref{bound-xi}), we
only obtain 
$$
x = \xi u^{-\frac{m+1}{2}} \geq C u^{-\frac{1}{2}}, 
$$
which is not sufficient
to claim that $x$ remains bounded as $u \to \infty$.

Using the transformation (\ref{scaling-explicit}) and
equations (\ref{equation-w}) and (\ref{equation-xi}), we obtain
the following system for $x$ and $z$ in the variable $u$:
\begin{equation}
\label{equation-x}
\frac{dx}{du} = \frac{1}{u} \left[ \frac{1}{u z} - \frac{m+1}{2} x \right]
\end{equation}
and
\begin{equation}
\label{equation-z}
\frac{dz}{du} = \frac{dx}{du} + \frac{1}{u^2} \left[ \frac{1}{u z} - \frac{m+3}{2} u z \right].
\end{equation}
It follows from the bound (\ref{bound-z}) and equation (\ref{equation-z}) that
there is a positive constant $C$ such that
$$
\frac{d}{du} (z-x) \geq - C u^{-\frac{3}{2}},
$$
for sufficiently large $u$. Since $u^{-\frac{3}{2}}$ is integrable as $u \to +\infty$,
then $z - x$ is bounded from below by a negative constant. Since $z$ is bounded
and approaches to zero as $u \to \infty$, we finally obtain
that $x$ is bounded from above by a positive constant for all sufficiently large $u$.
Finally, $(x_0,0,0)$ is an equilibrium state of system (\ref{infinity-dynamics}), therefore,
the trajectory approaches $(x_0,0,0)$ for some $x_0 \in [0,\infty)$.
\end{proof}

\vspace{0.25cm}

\begin{remark}
Integrating equation (\ref{equation-x}), we obtain
$$
x = \frac{c}{u^{\frac{m+1}{2}}} + \frac{1}{u^{\frac{m+1}{2}}} \int \frac{u^{\frac{m-1}{2}} du}{u z}.
$$
where $c$ is an arbitrary constant. If we can prove that $\lim_{u \to \infty} u z(u) = a_{\infty} > 0$, then
$$
\lim_{u\to\infty} x(u) = \frac{2}{a_{\infty} (m+1)} \in (0,\infty). 
$$
This would indicate that
the bound (\ref{bound-z}) is not sharp. However, we only have numerical data supporting this claim
for every $A_+ \in \mathbb{R}$.
\end{remark}

\vspace{0.25cm}

We will show numerically in \S5 that both pieces of the map (\ref{two-maps})
are monotonically increasing for all $A_+ \in \mathbb{R}$ and intersecting
at $x_0 = x_Q := \sqrt{2/(m+1)}$ for $A_+ = 0$.
Therefore, the piecewise $C^1$ map $\mathbb{R} \backslash \{0\} \ni A_+ \mapsto x_0 \in \mathbb{R}^+$
is in fact continuous at $A_+ = 0$. The exact value $x_Q$ corresponds to the
exact solution (\ref{exact-solution}) of the scalar differential equations (\ref{ode}).
This exact solution is recovered with the following elementary result.

\vspace{0.25cm}

\begin{proposition}
\label{lemma-exact}
There exists an exact solution of the dynamical systems (\ref{zero-dynamics})
and (\ref{infinity-dynamics}) given by
\begin{equation}
\label{exact-solution-zero}
\left\{ \begin{array}{l} \xi(\tau) =  \left( \frac{2^m (m+1)}{(m-1)^{m+1}} \right)^{\frac{1}{m-1}}
\left( \frac{a}{1 - \tau a} \right)^{\frac{m+1}{m-1}}, \\ u(\tau) =  \left( \frac{2(m+1)}{(m-1)^2} \right)^{\frac{1}{m-1}}
\left( \frac{a}{1 - \tau a} \right)^{\frac{2}{m-1}}, \\
w(\tau) = \left( \frac{2^m (m+1)}{(m-1)^{m+1}} \right)^{\frac{1}{m-1}}
\left( \frac{a}{1 - \tau a} \right)^{\frac{m+1}{m-1}}, \end{array} \right. \quad \tau \in (-\infty,a^{-1}),
\end{equation}
and
\begin{equation}
\label{exact-solution-infinity}
\left\{ \begin{array}{l} x(s) = x_Q, \\ y(s) =  \frac{b x_Q^{-1}}{1+ s b},  \\
z(s) = \frac{b}{1+ s b}, \end{array} \right. \quad s \in (-b^{-1},\infty),
\end{equation}
respectively, where $x_Q = \sqrt{2/(m+1)}$, whereas $a$ and $b$ are arbitrary positive parameters.
\end{proposition}

\vspace{0.25cm}

\begin{proof}
Consider the system (\ref{infinity-dynamics}) and try the reduction $x = x_0$, where $x_0$ is constant in $s$.
Then, on setting $y = \frac{m+1}{2} x_0 z$ we obtain the following differential equations:
$$
y' = - zy, \quad z' = y^2 - \frac{m+3}{2} z^2.
$$
This system is compatible with the constraint $y = \frac{m+1}{2} x_0 z$ if $z' = -z^2$
and $x_0^2 = \frac{2}{m+1} \equiv x_Q^2$, so that $y = x_Q^{-1} z$. The general solution of $z' = -z^2$ is
$z(s) = \frac{b}{1+ s b}$ for a positive parameter $b$. The solution is defined for $s > -b^{-1}$.

Consider the system (\ref{zero-dynamics}) and try the reduction $\xi = w = \dot{u}$. Then, the system
is compatible with the reduction if $\dot{\xi} = u^m$ and $u^{m+1} = \frac{m+1}{2} \xi^2$.
Therefore,
$$
\dot{\xi} = \left( \frac{m+1}{2} \right)^{\frac{m}{m+1}} \xi^{\frac{2m}{m+1}},
$$
which admits a general solution given by
$$
\xi(\tau) =  \left( \frac{2^m (m+1)}{(m-1)^{m+1}} \right)^{\frac{1}{m-1}}
\left( \frac{a}{1 - \tau a} \right)^{\frac{m+1}{m-1}},
$$
where $a$ is an arbitrary positive constant. The solution exists
for $\tau < a^{-1}$. Other components $u$ and $w$ are found from
the above relations.
\end{proof}

\vspace{0.25cm}

\begin{remark}
Note that in the exact solution (\ref{exact-solution-zero}), $\xi(\tau) \to 0 = A_+$ as $\tau \to -\infty$. Therefore,
the exact solution of Proposition \ref{lemma-exact} corresponds to the choice $A_+ = 0$ in Lemma
\ref{lemma-connection}.
\end{remark}

\subsection{\label{H-minus-system}The system for $H_-$ ($t < 0$)}

By Propositions \ref{proposition-center-manifold-zero} and \ref{proposition-stable-manifold-zero},
for every nonzero $A \equiv A_-$, there is a unique trajectory
of the dynamical system (\ref{zero-dynamics}) in variables $(\xi,u,w)$ that departs from
the equilibrium point $(A_-,0,0)$ as $\tau \to -\infty$ and belongs to the domain $u > 0$.
This trajectory is contained in the unstable manifold $W_u(A_-,0,0)$ if $A_- > 0$ and
in the center manifold $W_c(A_-,0,0)$ if $A_- < 0$.

By Propositions \ref{proposition-center-manifold-infinity} and \ref{proposition-stable-manifold-infinity},
for every $x_0 > 0$, there is a one-dimensional set of trajectories
of the dynamical system (\ref{infinity-dynamics}) in variables $(x,y,z)$
that reaches the equilibrium point $(x_0,0,0)$ as $s \to +\infty$ and belongs
to the domain $y > 0$. This trajectory is contained in the center manifold $W_c(x_0,0,0)$.

If we try an argument used in the proof of Lemma \ref{lemma-connection}, then it becomes
clear that most of the trajectories departing  from
the equilibrium point $(A_-,0,0)$ in system (\ref{zero-dynamics})
will not arrive to the  equilibrium point $(x_0,0,0)$ in system (\ref{infinity-dynamics})
but instead reach the value $u = 0$ in a finite $\tau \in \mathbb{R}$.
This indicates that, first, there are very few values of $A_-$, for which the
trajectories may reach infinite values for $u$, and second, the numerical method
should not be based on the trajectories departing from the equilibrium point $(A_-,0,0)$
(such a shooting method was used previously in \cite{Foster}). Instead, it may be better
to look for the one-dimensional trajectory departing the  equilibrium point $(x_0,0,0)$ in
system (\ref{infinity-dynamics}) in the negative direction of the time variable $s$.

To illustrate the previous point, we show on figure \ref{return-to-zero}
the trajectories of the system (\ref{zero-dynamics}) with $m = 3$ starting from the equilibrium
point $(A_{\pm},0,0)$ along either center (for $A_+ > 0$) or unstable (for $A_- > 0$) manifolds.
The trajectories of the system for $H_+$ extend from small to infinite values of $H_+$, see panel (a).
Contrastingly, the trajectories of the system for $H_-$ turn back and return to small values of $H_-$, see panel (b).
Note that the return time is significantly different between the first two and the last two trajectories.
This indicates the presence of a particular value of $A_-$, for which there exists a trajectory
that extends from small to infinite values of $H_-$, see (\ref{m-is-3-vals}) below.

\begin{figure}[h!]        \centering
\includegraphics[width=0.49\textwidth]{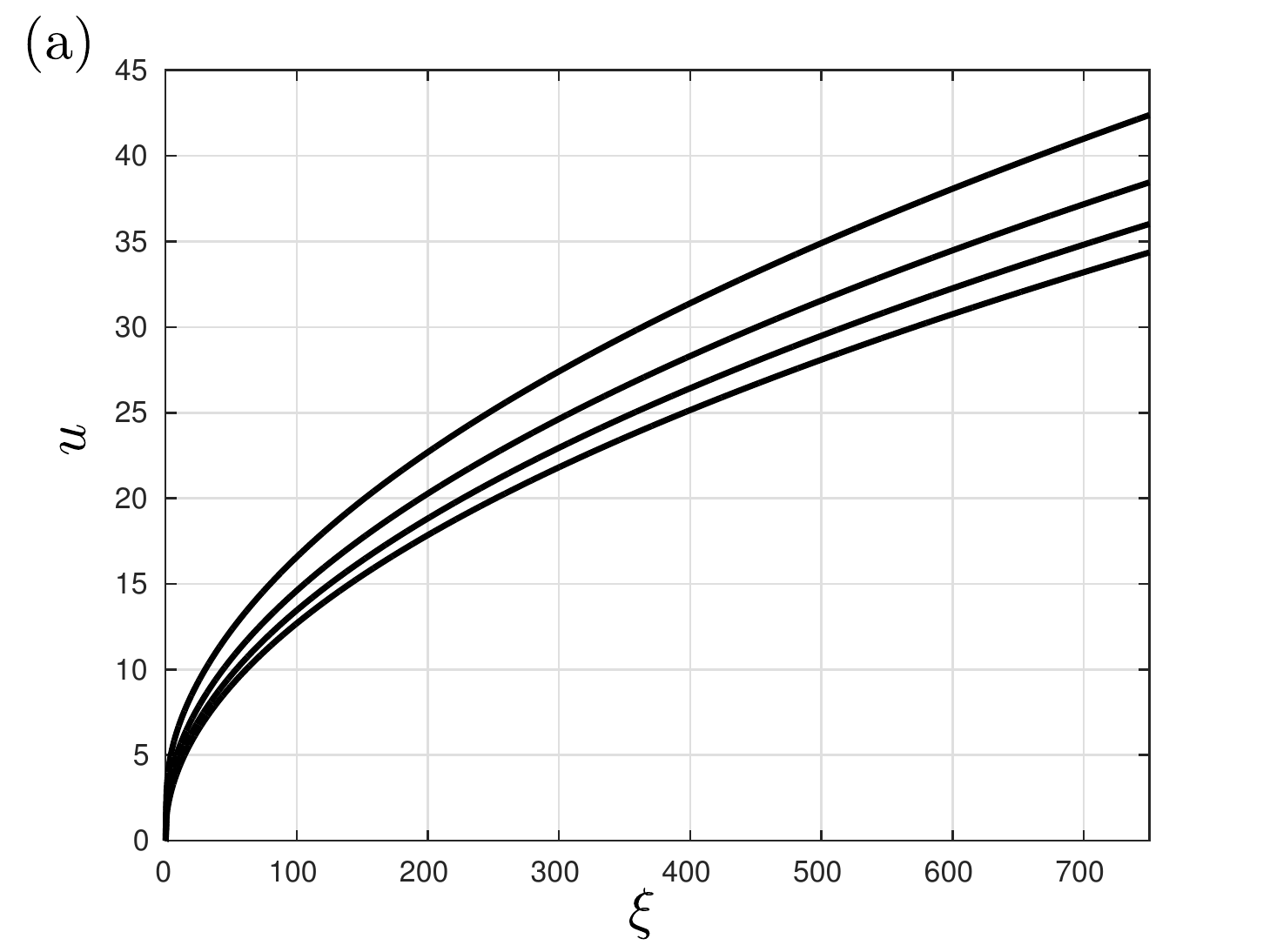}
\includegraphics[width=0.49\textwidth]{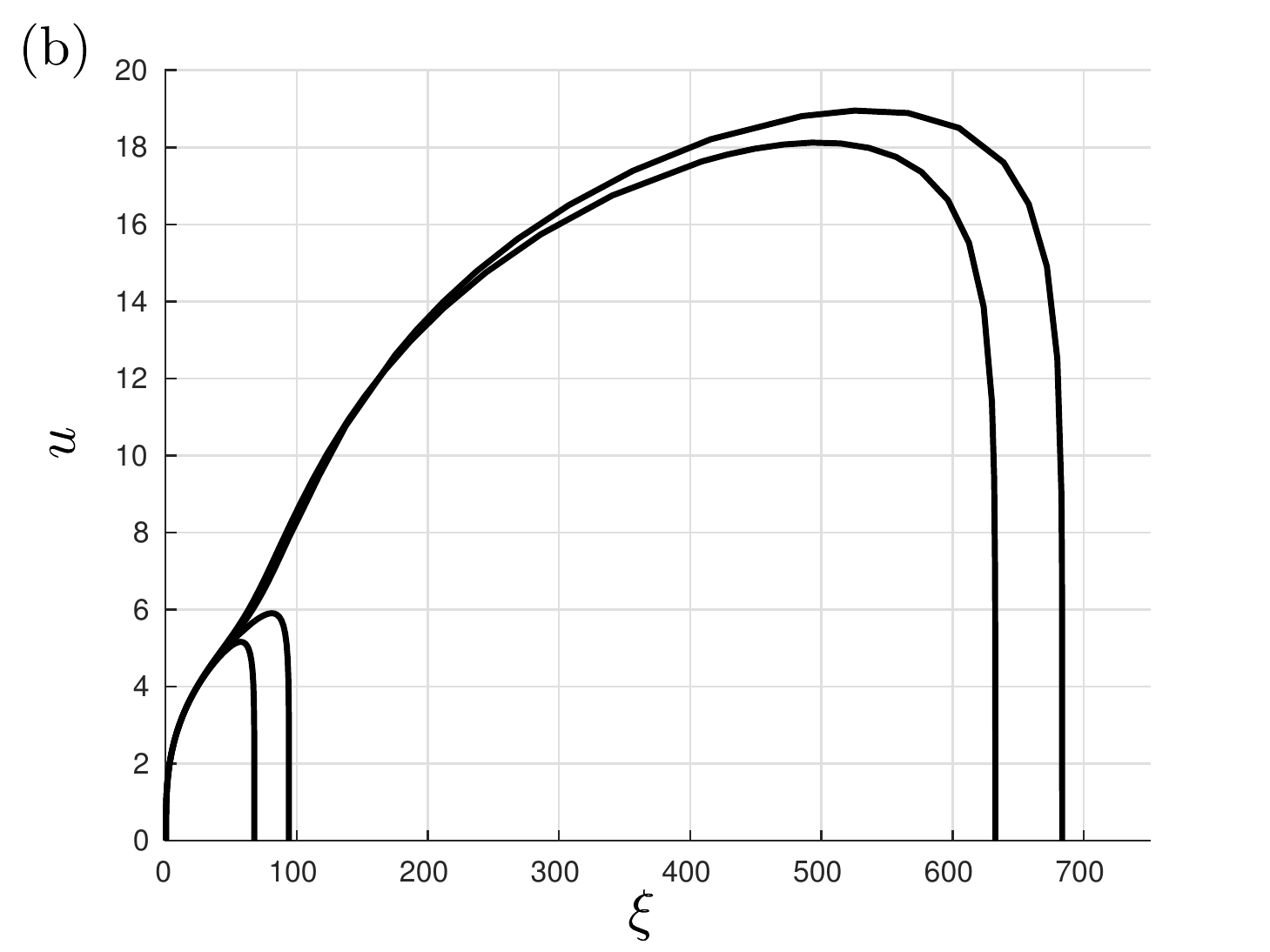}
\caption{Panels (a) and (b) show trajectories of the system (\ref{zero-dynamics}) with $m = 3$ for $H_+$ and $H_-$ respectively.
In both cases, the trajectories start from the equilibrium point $(A_{\pm},0,0)$ and depart along either the center
or the unstable manifolds.}
\label{return-to-zero}
\end{figure}
% Code: /home/jmfoster/Documents/Shooting_for_reversing_solutions/Pelinovsky_shooting/automated_shooting_scheme/automated/auto_track_semi_automatic/auto_track_butchered.m
% Code: /home/jmfoster/Documents/Shooting_for_reversing_solutions/Pelinovsky_shooting/automated_shooting_scheme/automated/auto_track_semi_automatic/auto_track_butchered_m_is_2.m

In order to justify our numerical scheme, we shall prove that the trajectory departing the equilibrium
point $(x_0,0,0)$ in system (\ref{infinity-dynamics}) in the negative $s$ direction either intersects
the plane $u = 0$ or the plane $w = 0$ of system (\ref{zero-dynamics}).\\

\begin{lemma}
\label{lemma-continuation}
Fix $x_0 > 0$ and consider a one-parameter trajectory of the dynamical system (\ref{infinity-dynamics}) for
the lower sign such that $(x,y,z) \to (x_0,0,0)$ as $s \to +\infty$ and $y > 0$. Then, there exists
an $s_0 \in \mathbb{R}$ (or $s_0 = -\infty$) such that
\begin{itemize}
\item[(i)] either $z(s_0) = 0$ and $y(s_0) \in (0,\infty]$,
\item[(ii)] or $y(s) \to +\infty$ as $s \to s_0$, whereas
$\lim\limits_{s \to s_0} \frac{z(s)}{y(s)^{\frac{m+3}{2}}} \in [0,\infty)$.
\end{itemize}
\end{lemma}

\vspace{0.25cm}

\begin{proof}
For convenience, let us reverse the time variable, by transforming $s \to -s$, and rewrite system
(\ref{infinity-dynamics}) with the lower sign in the negative direction of $s$:
\begin{equation}
\label{infinity-dynamics-negative-time}
\left\{ \begin{array}{l}
x' = \frac{m+1}{2} x z - y, \\
y' = z y, \\
z' = y (1 - y) - \frac{m+1}{2} x z + \frac{m+3}{2} z^2.\end{array} \right.
\end{equation}
By Proposition \ref{proposition-center-manifold-infinity}, we have $y > 0$ and $z > 0$
for the trajectory departing from the equilibrium point $(x_0,0,0)$ (in negative $s$)
along $W_c(x_0,0,0)$. From the second equation of system (\ref{infinity-dynamics-negative-time}),
$y$ remains an increasing function of negative $s$ as long as $z$ remains positive.
Therefore, we have an alternative: either $z$ vanishes before $y$ diverges or
$y$ diverges before $z$ vanishes.

The first choice of the alternative gives case (i). For the second choice,
let us consider variables $\xi$ and $w$ given by (\ref{scaling-explicit})
and parameterized by $y$ in the limit $y \to +\infty$
(the map $s \mapsto y$ is one-to-one and onto). By dividing the first and third equations
in system (\ref{infinity-dynamics-negative-time}) by the second equation, we obtain
$$
\frac{dx}{dy} = \frac{m+1}{2} \frac{x}{y} - \frac{1}{z}, \quad
\frac{dz}{dy} = \frac{1-y}{z} - \frac{m+1}{2} \frac{x}{y} + \frac{m+3}{2} \frac{z}{y}.
$$
By using variables $\xi$ and $w$ given by (\ref{scaling-explicit}), we obtain
\begin{equation}
\label{closed-infinity}
\frac{d\xi}{dy} = - \frac{1}{y^{m+2} w}, \quad
\frac{dw}{dy} = \frac{1-y}{y^{m+3} w} - \frac{m+1}{2} \frac{\xi}{y^2}.
\end{equation}
To show that the second choice of the alternative above gives case (ii), we will
prove that $w$ remains finite as $y \to +\infty$. This is done by a contradiction.
Assume that $w \to +\infty$ as $y \to +\infty$. Therefore, there exists $w_0 > 0$ such that
$w \geq w_0$ for all sufficiently large $y$. Then, from the first equation of system (\ref{closed-infinity}),
we have for sufficiently large $y$,
$$
\left| \frac{d\xi}{dy} \right| \leq  \frac{1}{y^{m+2} w_0}.
$$
Because $y^{-m-2}$ is integrable at infinity, there is a finite positive $\xi_{\infty}$ such that
$|\xi| \leq \xi_{\infty}$ for all sufficiently large $y$. Then, from the second equation of
system (\ref{closed-infinity}), we obtain for sufficiently large $y$,
$$
\left| \frac{dw}{dy} \right| \leq \frac{y-1}{y^{m+3} w_0} + \frac{m+1}{2} \frac{\xi_{\infty}}{y^2}.
$$
Since both $y^{-m-2}$ and $y^{-2}$ are integrable at infinity, there is a finite positive
$w_{\infty}$ such that $w \leq w_{\infty}$ for all sufficiently large $y$, contradicting the
assumption that $w \to +\infty$ as $y \to +\infty$. Therefore, the case (ii) is proved.
\end{proof}

\vspace{0.25cm}

\begin{corollary}
The trajectory of system (\ref{infinity-dynamics})
departing from the equilibrium point $(x_0,0,0)$ with $x_0 > 0$ in the negative direction of $s$
intersects either the half-plane $w = 0$ and $u \geq 0$ in system (\ref{zero-dynamics}) in case (i)
of Lemma \ref{lemma-continuation} or the half-plane $u = 0$, $w \geq 0$ in system (\ref{zero-dynamics})
in case (ii). Moreover, in the corresponding cases,
\begin{itemize}
\item[(i)] if $\lim\limits_{s \to s_0} u(s) > 0$, then $\lim\limits_{s \to s_0} |\xi(s)| < \infty$
\item[(ii)] if $\lim\limits_{s \to s_0} w(s) > 0$,
then $\lim\limits_{s \to s_0} |\xi(s)| < \infty$.
\end{itemize}
Consequently, one can define two piecewise $C^1$ maps
\begin{equation}
\label{maps-negative}
\mbox{\rm (i)} \quad  \mathbb{R}^+ \ni x_0 \mapsto (\xi,u) \in \mathbb{R} \times \mathbb{R}^+ \quad \mbox{\rm and} \quad
\mbox{\rm (ii)} \quad \mathbb{R}^+ \ni x_0 \mapsto (\xi,w) \in \mathbb{R} \times \mathbb{R}^+.
\end{equation}
\label{corollary-continuation}
\end{corollary}

\begin{proof}
In case (i), it is trivial to see that $z(s_0) = 0$ and $y(s_0) \in (0,+\infty]$ corresponds to the half-plane $w = 0$
and $u \geq 0$. The first equation of system (\ref{infinity-dynamics-negative-time}) can be written for the variable $\xi$ as follows:
$$
\xi' = -y^{\frac{1-m}{2}},
$$
where the prime still denotes the derivative with respect to the time variable $s$ in the negative
direction of $s$. If $y$ remains finite as $s \to s_0$ (so that $u(s_0) > 0$),
then $\xi(s_0)$ is bounded.

In case (ii), it is also trivial to see that $\lim_{s \to s_0} y(s) = +\infty$ and
$\lim\limits_{s \to s_0} z(s) y(s)^{-\frac{m+3}{2}} \in [0,\infty)$ corresponds
to the half-plane $u = 0$ and $w \geq 0$. If $w$ remains nonzero in the limit $y \to +\infty$,
then, there exists $w_0 > 0$ such that $w \geq w_0$ for all sufficiently large $y$.
Then, the same analysis as in Lemma \ref{lemma-continuation} applies to the first equation of system (\ref{closed-infinity})
and shows that $\xi$ remains finite as $y \to +\infty$.
\end{proof}

\vspace{0.25cm}\emph{}

Unfortunately, we do not control the value of $\xi$ at the intersection of
the two piecewise $C^1$ maps (\ref{maps-negative}). However,
we will show numerically in \S\ref{num-scheme} that the piecewise $C^1$ maps (\ref{maps-negative})
are typically connected at the points where $u = w = 0$ and $\xi = A \in \mathbb{R}$. In this case,
a true solution $H_-$ of the differential equation
(\ref{ode}) with the lower sign satisfying properties (\ref{(i)})--(\ref{(iii)}) exist.

\section{\label{num-scheme}Numerical results}

Let us describe a new numerical approach, based on the results of Lemmas \ref{lemma-connection}
and \ref{lemma-continuation}, that will be used to furnish meaningful solutions to 
the  differential equations (\ref{ode}), for $H_-$ and $H_+$. In \S\ref{t-neg-shooting} and \S\ref{t-pos-solutions} we describe the numerical procedures for finding solutions for $H_-$ and $H_+$ respectively. Finally, in \S\ref{num-sum},
we summarize the results of the numerical experiments and compare them with the results
found in Foster \emph{et. al.} \cite{Foster}.

\subsection{\label{t-neg-shooting}Solutions for $H_-$ ($t<0$)}

As discussed in \S\ref{H-minus-system}, we wish to numerically construct a unique
trajectory from infinite to finite values of $H_-$. To do so, we integrate the system (\ref{infinity-dynamics}) from near the equilibrium point $(x_0,0,0)$
in the far-field towards the equilibrium point $(A_-,0,0)$ of the system (\ref{zero-dynamics})
in the near-field. The numerical procedure is carried out as follows:

\begin{remunerate}
\item Select a value of $x_0>0$. Ideally, one would like to begin by using this choice of $x_0$
to specify unique initial values for $(x,y,z)=(x_0,0,0)$, and then numerically integrating
the system (\ref{infinity-dynamics}) backward in the `time' variable $s$. Equivalently,
one could integrate the system (\ref{infinity-dynamics-negative-time}) forwards in time. However,
since $(x_0,0,0)$ is an equilibrium point, it is not possible to escape $(x_0,0,0)$ in a finite time.
Thus, in order to ensure that any numerical integration scheme can depart from near
the equilibrium point along the center manifold, $W_c(x_0,0,0)$, it is necessary to take
a `small step', say $\delta \ll 1$, away from $(x_0,0,0)$ using the relevant asymptotic behaviour.
Using (\ref{center-manifold-infinity}) and (\ref{center-manifold-infinity-dynamics}) we find
that a trajectory on the center manifold $W_c(x_0,0,0)$ has the local behaviour
\begin{equation}
\label{numeric-begin-t-negative}
\left\{ \begin{array}{l}
x = x_0 + \left( \frac{m+1}{2} - \left( \frac{m+1}{2} \right)^2 x_0^2 \right) \delta + \mathcal{O}(\delta^2), \\
y = \frac{m+1}{2} x_0 \delta + \mathcal{O}(\delta^2), \\
z = \delta, \end{array} \right.
\end{equation}
for small positive values of $\delta$.
Having selected values for both $x_0$ and $\delta$, the behaviours (\ref{numeric-begin-t-negative})
may be used to specify unique (pseudo-)initial values for $(x,y,z)$ and to begin the numerical
integration of the system (\ref{infinity-dynamics}) in the direction of decreasing $s$. In this study,
numerical integration of the system (\ref{infinity-dynamics}) was carried out using the {\tt ode45} routine
in MATLAB with the default settings, except {\tt AbsTol} and {\tt RelTol} which were both set to have
a value of $10^{-10}$. Selecting an appropriate value of $\delta$ is a somewhat ad-hoc procedure:
there is trade-off between taking $\delta$ too small, which renders it difficult for the numerical
integration to escape the neighbourhood of the equilibrium point (leading to poor accuracy of the integration), and taking $\delta$ too large
which could result in low accuracy of the asymptotic expansion (\ref{numeric-begin-t-negative}).
However, we found that choosing $\delta \in (10^{-3},10^{-2})$  gave good results over the ranges
of parameters we studied. Robustness of the results with respect to changes in: (i) the choice of $\delta$, and; (ii) the number of terms in the asymptotic expansion (\ref{numeric-begin-t-negative}) were verified.

\item The result of Corollary \ref{corollary-continuation} asserts that all such trajectories will ultimately
--- in either a finite or infinite time --- intersect with either the plane $w=0$ or $u=0$, see
the maps (\ref{maps-negative}). To ensure accurate numerical integration of the system in the
near-field, in variables $(\xi,u,w)$, it is necessary to `switch' from integrating the
far-field system (\ref{infinity-dynamics}) to the near-field system (\ref{zero-dynamics}) backward in time.
The choice of conditions under which this switch should occur is, again, somewhat arbitrary. However,
as long as the values of $(x,y,z)$ --- and hence the values of $(\xi,u,w)$ --- are all finite and non-zero,
this switching is valid at any point. In this study, we chose to switch from integrating
(\ref{infinity-dynamics}) to (\ref{zero-dynamics}) when $x y^{-(m+1)/2} = 20$ (or equivalently
when $\xi = 20$). However, we verified that our results were robust to changes in the choice
of switching conditions. This switching procedure can be readily automated within MATLAB using
the {\tt Events} function to autonomously: (i) stop the integration of the system (\ref{infinity-dynamics})
when specified the conditions are satisfied; (ii) read-off the final values of $(x,y,z)$;
(iii) transform these to corresponding values for $(\xi,u,w)$ using the change of variables
(\ref{scaling-explicit}), and; (iv) begin the integration of (\ref{zero-dynamics}) backwards
in time from the appropriate initial data.

\item The integration of the system (\ref{zero-dynamics}) is then continued backward
in the time variable $\tau$ until either $w=0$ or $u=0$. Again, we used the {\tt Events}
function to autonomously detect when either of these events occurred and to stop the
integration. It is noteworthy that we found it helpful to use {\tt ode15s}
to integrate the near-field system --- again, the default settings were used
with the exception of {\tt AbsTol} and {\tt RelTol} which we both set to $10^{-10}$.
Although {\tt ode15s} is typically slower than {\tt ode45}, it is considerably
more appropriate to deal with integrating systems of equations that exhibit apparent stiffness.
This apparent stiffness, manifested as rapid changes in the direction of the trajectory
in variables $(\xi,u,w)$, arises as an artifact of the infinite time required to
reach the equilibrium point $(A_-,0,0)$. Thus, if a trajectory approaches very
close to the equilibrium point it appears to be rapidly rejected from that neighbourhood.
When the integration is terminated, we record the following pieces of data: 
(i) the selected value of $x_0$; (ii) whether the trajectory reached $u=0$ or $w=0$, and; 
(iii) the value of either $(\xi,u)$ or $(\xi,w)$ at the termination point. 
These data define a point on one of the two piecewise $C^1$ maps defined in (\ref{maps-negative}). 
It is by computing a large number of trajectories, each emanating from different values of $x_0$, 
that we are able to trace out the forms of these maps.
\end{remunerate}

The non-trivial solutions for $H_-$ that we seek correspond to trajectories emanating from particular
equilibrium points in the far-field, say $(x_0^*,0,0)$, that reach the near-field
equilibrium point, $(A_-,0,0)$ for some finite $A_- \neq 0$. In addition to
these non-trivial solutions, we recall that for all values of $m>1$ there exists a trivial solution
given by (\ref{exact-solution-zero}) and (\ref{exact-solution-infinity}) that emanates from
$x_0^* = x_Q = \sqrt{2/(m+1)}$ and reaches $A_- =0$, as discussed in Proposition \ref{lemma-exact}.
Some representative results for $m=2,3,4$ and $5$ are shown in figure \ref{individual-maps}. In these plots,
a suitable non-trivial solution for $H_-$ is found by identifying a value of $x_0=x_0^*$ for
which the value of $(\xi,u)=(A_-,0)$ --- or $(\xi,w)=(A_-,0)$ --- at the termination point.

In figure \ref{example-shooting}, we show some representative computations
highlighting the differences between the near-field behaviour of trajectories
local to a true solution with $A_-<0$ and $A_->0$. As is evidenced by figure \ref{individual-maps}, close to a trajectory with $A_->0$
the piecewise continuous $C^1$ maps defined in (\ref{maps-negative}) are smooth,
whereas close to a trajectory with $A_-<0$ the maps exhibit rapid changes and (vertical) cusp-like features. As a result, determining negative value(s) of $A_-$ is more challenging ---
despite resolving $x_0$ to machine precision (approximately $10^{-14}$ in the standard IEEE double precision),
it is not possible to approximate the value of $x_0^*$ sufficiently well that an accurate
estimate of $A_-(<0)$ can be determined. In these cases, we therefore found it necessary
to implement one additional stage in the numerical scheme as follows.

Once $x_0^*$ had been determined up to machine precision, and two `limiting'
trajectories had been identified (one emanating from $x_0^* \pm \hat{\delta}$ and
terminating at $u=0$, and the other emanating from $x_0^* \mp \hat{\delta}$ and
terminating at $w=0$, where $\hat{\delta} \ll 1$) the expected near-field
linear asymptotic behaviour of $u(\xi)$, according to the expansion (\ref{center-manifold-zero-asymptotics})
for the solution $H_-$ in Theorem \ref{theorem-center-manifold-zero}, is clearly visible.
This linear behaviour can then be extrapolated, in the direction of decreasing $\xi$,
until it intersects the $\xi$-axis. This intersection point is, to a good approximation,
the value of $A_-(<0)$ corresponding to the trajectory emanating from $x_0^*$.
Panel (a) of figure \ref{example-shooting} gives an example of the linear extrapolation procedure described above.

Using the procedure described above, for all values of $m>1$, we recover the trivial solution discussed in Proposition \ref{lemma-exact}. In addition, we see that in the case $m=3$ there is only one suitable non-trivial solution with the following data:
\begin{equation}
m = 3 : \quad x_0^* \approx 0.767, \quad A_- \approx 0.129. \nonumber
\end{equation}
For $m=4$ and $m=5$, similar results are observed with one trivial and only one non-trivial solution as follows:
\begin{equation}
m=4 : \quad x_0^* \approx 1.165, \quad A_- \approx 0.386, \nonumber
\end{equation}
and
\begin{equation}
m=5 : \quad x_0^* \approx 1.666, \quad A_- \approx 0.501. \nonumber
\end{equation}
Contrastingly, in the case $m=2$ we find that three suitable non-trivial solutions exist with the data:
\begin{equation} \nonumber
m=2: \quad \left\{ \begin{array}{l}
 x_0^* \approx 0.338, \quad A_- \approx -2.804,  \\*[1mm]
 x_0^* \approx 0.137, \quad A_- \approx -0.932, \\*[1mm]
 x_0^* \approx 0.0592, \quad A_- \approx -0.546. \end{array}  \right.
\end{equation}
Notably all values of $A_-$ for $m=2$ are negative, whereas for $m=3,4$ and $5$ they are positive.

In addition to the detailed results for $m=2,\, 3, \, 4$ and $5$ shown in figure \ref{individual-maps},
we also show in figure \ref{all-sols-negative-t} the values of $x_0^*$ determined
for all values of $m$ from $m = 1$ to $m = 8$. For the primary red and blue branches, emanating from $m=3$
along the black branch, we show the corresponding values of $A_-$ in figure \ref{negative-t-A-minus}.
Intriguingly, the numerical results indicate that in addition to the exact solution -- which is valid
for all $m>1$ --- there are a whole host of additional solutions, some with $A_->0$, and others
with $A_-<0$. In particular, there is at least one additional trajectory corresponding to a
suitable solution for $H_-$ with $A_->0$ for all values of $m \gtrsim 2.978$. Further,
for all $m<3$ there exists at least one additional solution for $H_-$, although, in this case
for a value of $A_-<0$. Another noteworthy feature of the plots shown in panels (a)-(d) of
figure \ref{all-sols-negative-t} is that at each value of $m=(2N-1)$ for $N \in {\mathbb N}$
additional branches of solutions depart from the branch along which $x_0 = x_Q$ and $A_-=0$.
The underlying reason for this structure is as yet not understood.

\begin{figure}          \centering
\includegraphics[width=0.49\textwidth]{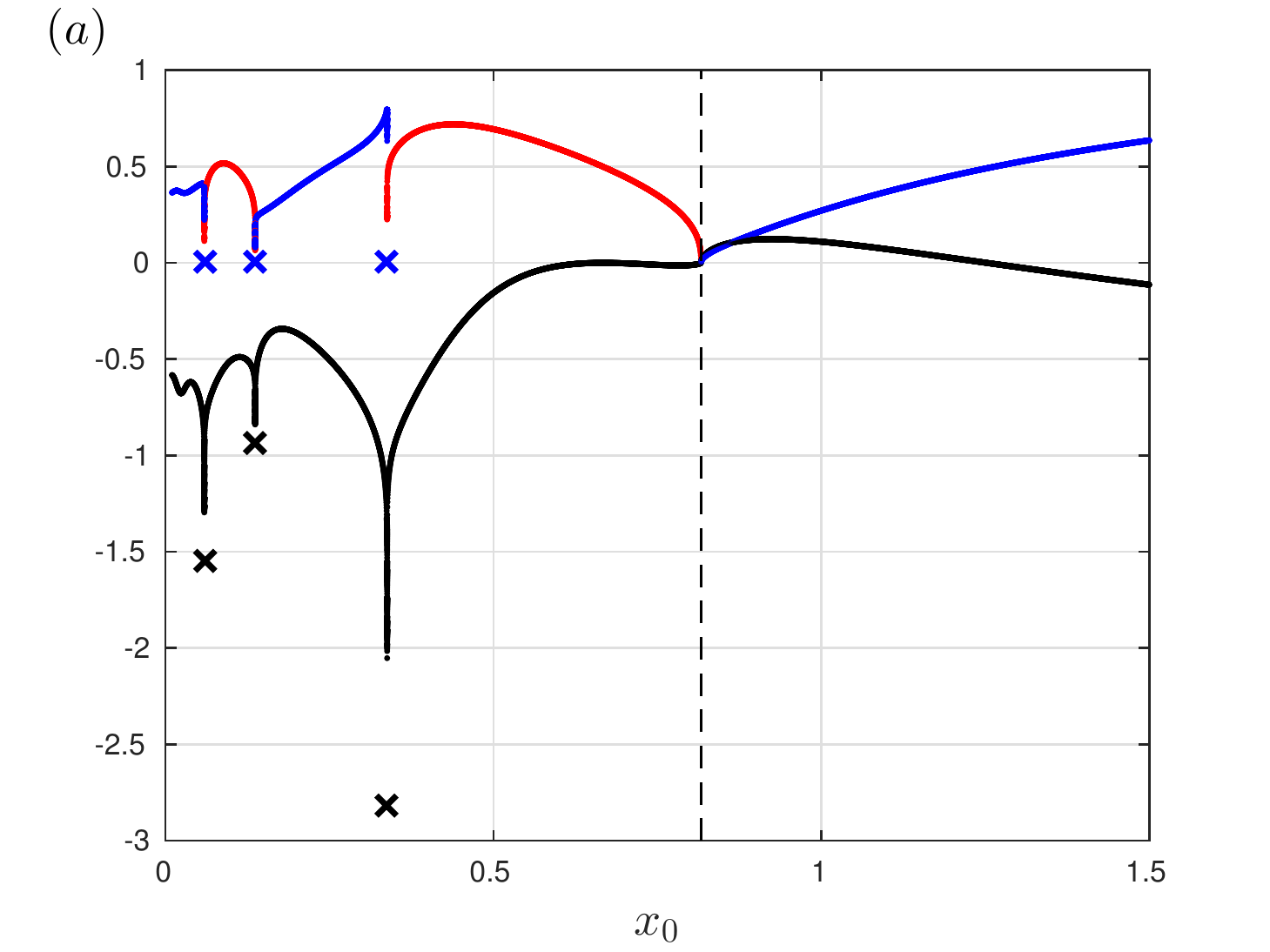}
\includegraphics[width=0.49\textwidth]{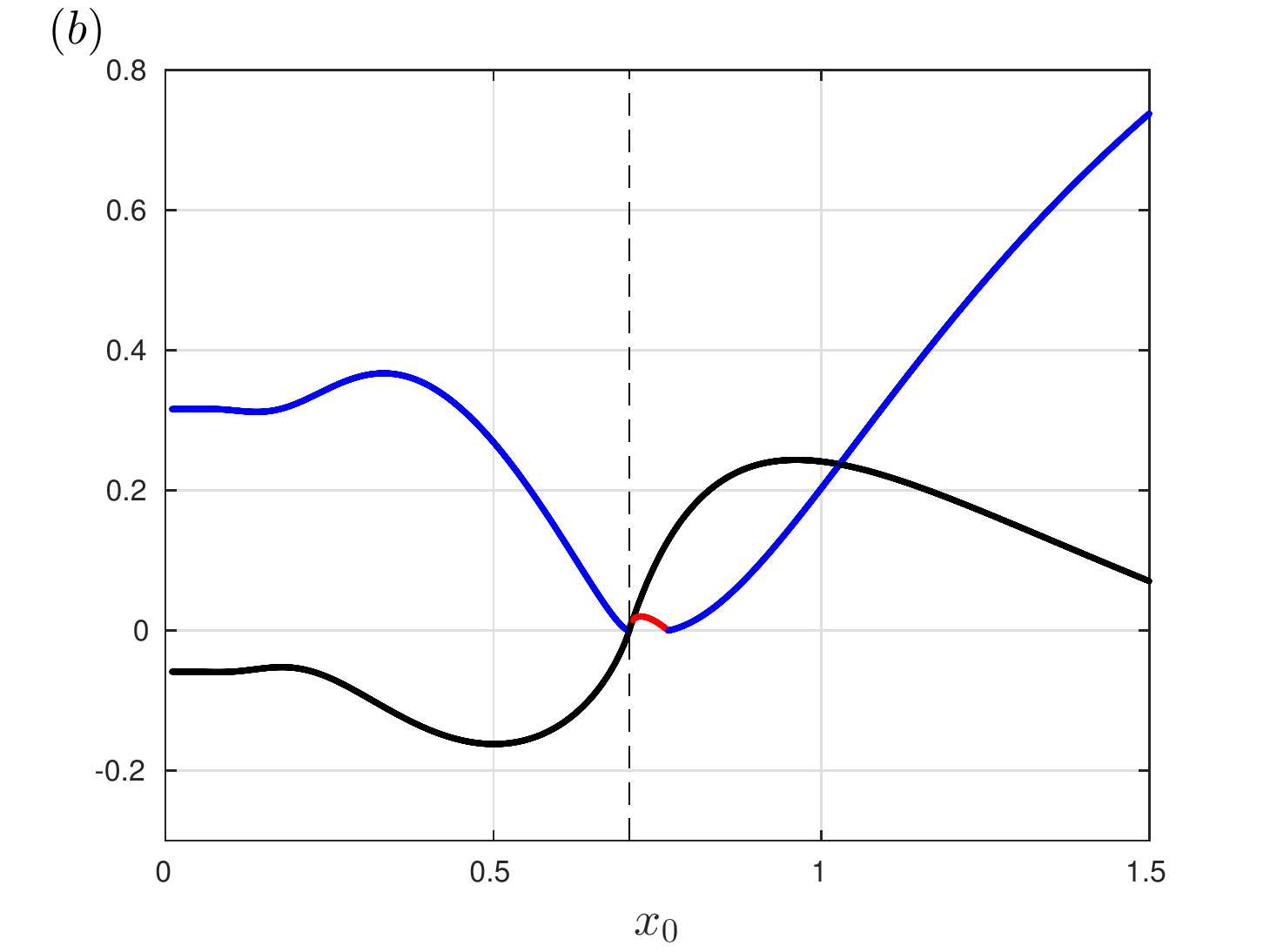}
\includegraphics[width=0.49\textwidth]{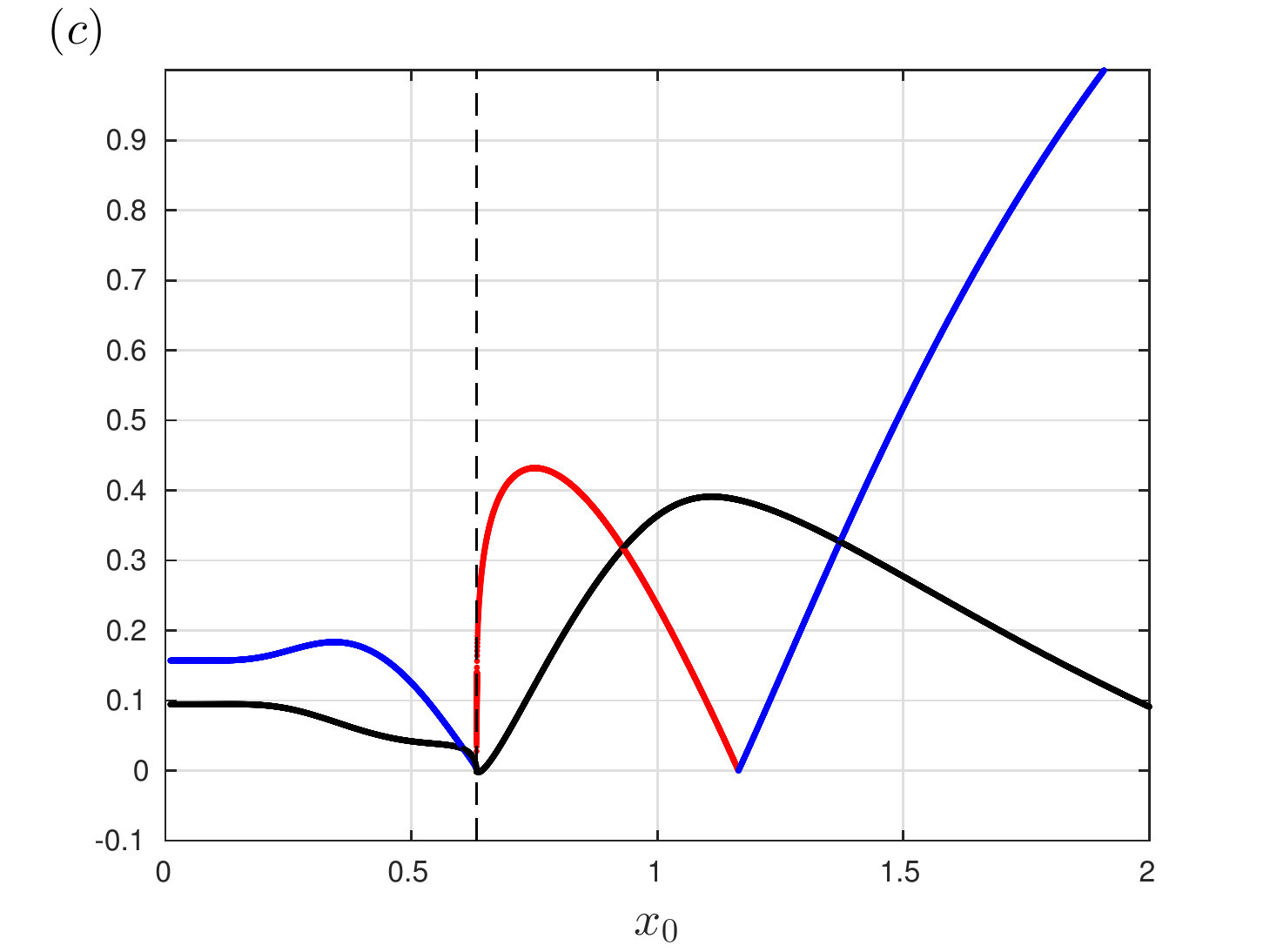}
\includegraphics[width=0.49\textwidth]{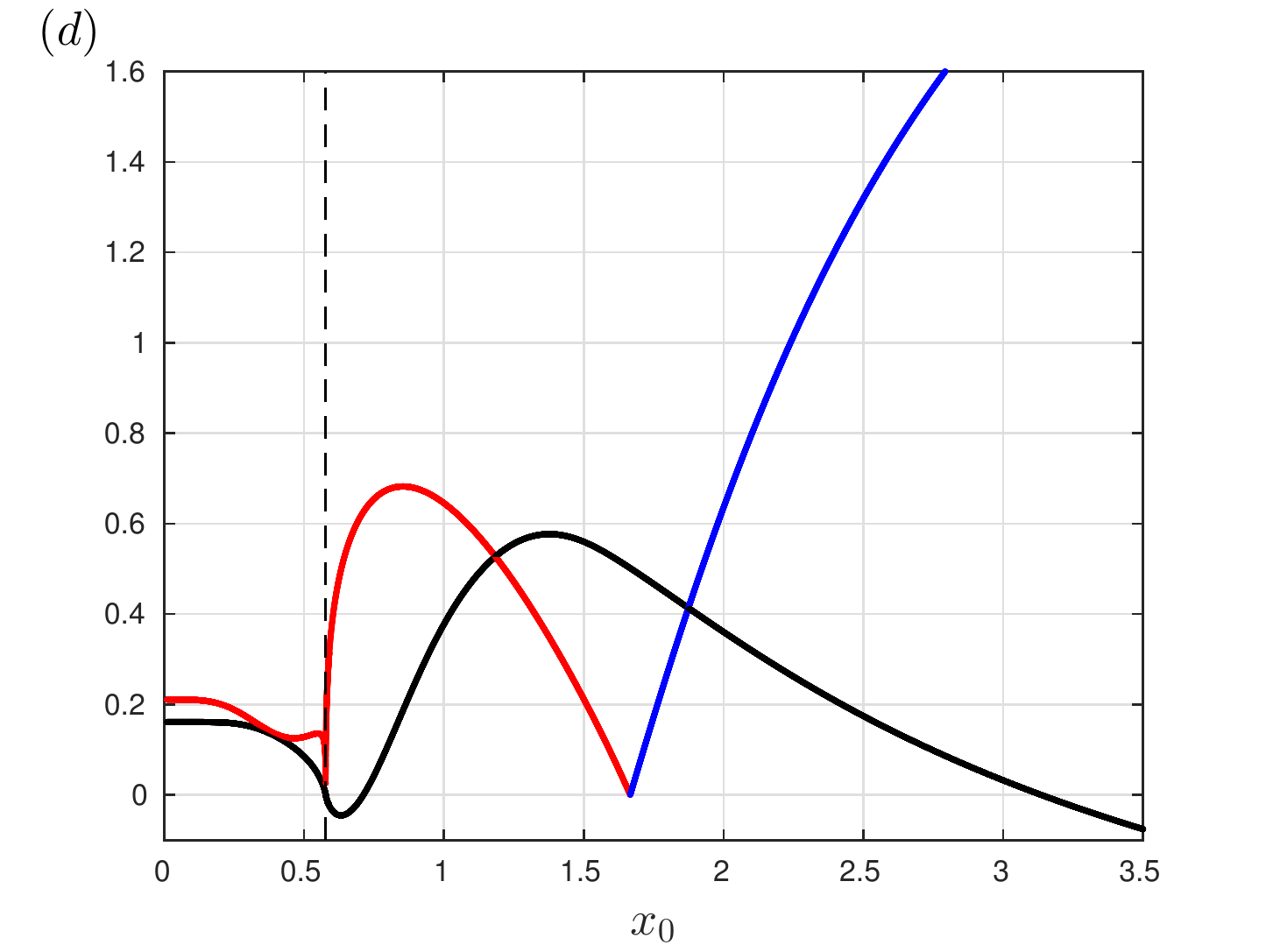}
\caption{Panels (a)-(d) show plots of the piecewise $C^1$ maps defined in (\ref{maps-negative})
for $m=2,3,4$ and $5$ respectively. In all cases the blue, red and black curves show the value
of $w$ at $u=0$, the value of $u$ at $w=0$ and the value of $\xi$ at the termination point
respectively. The dashed vertical line indicates the value of $x_0=x_Q$ corresponding to
the exact solution (\ref{exact-solution}). The crosses on panel (a) mark the data points
extracted using the extrapolation procedure discussed in the text.}
\label{individual-maps}
\end{figure}
% Code: /home/jmfoster/Documents/Shooting_for_reversing_solutions/Pelinovsky_shooting/automated_shooting_scheme/automated/auto_track_semi_automatic/red_blue_black.m

\begin{figure}          \centering
\includegraphics[width=0.49\textwidth]{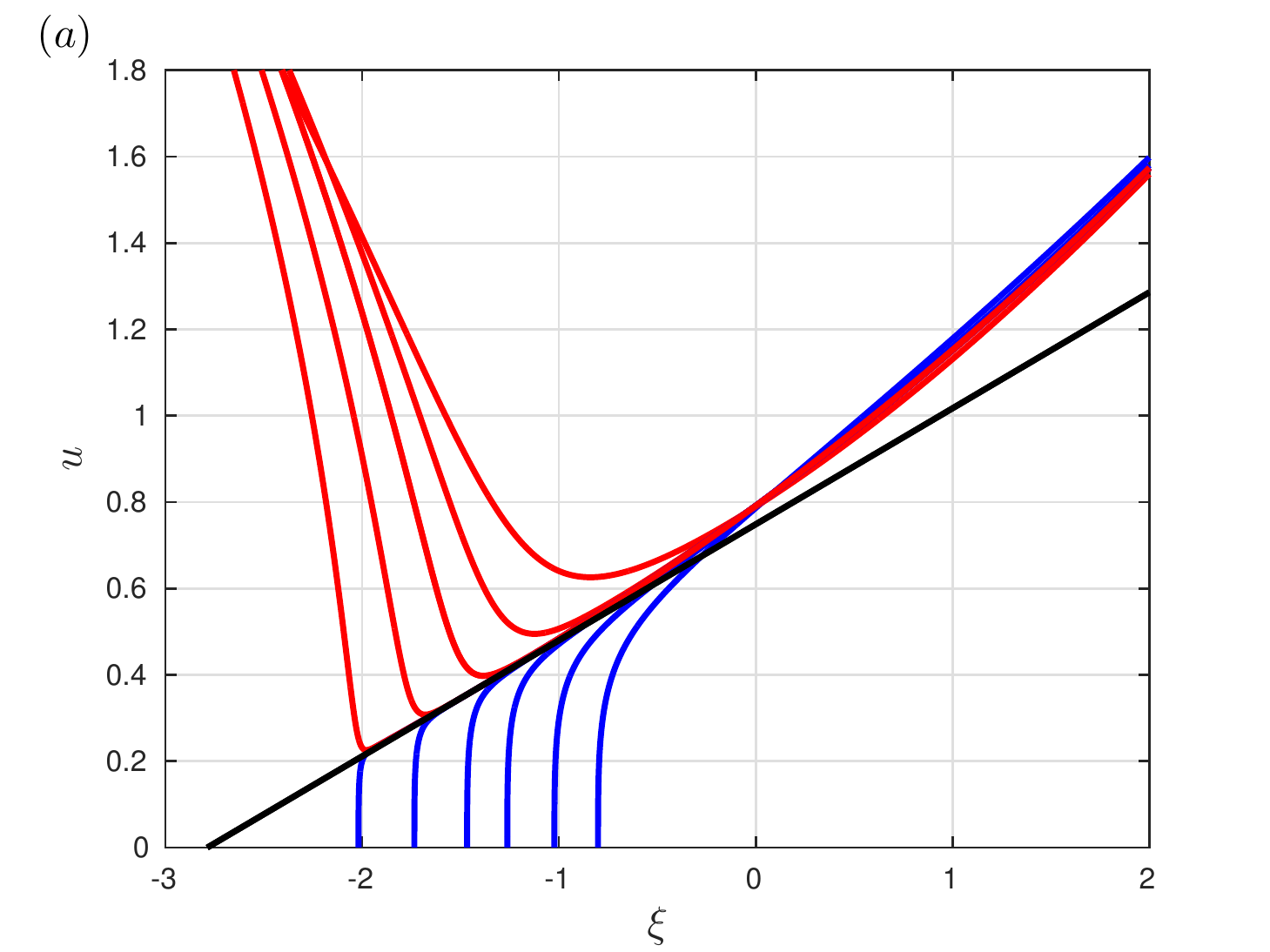}
\includegraphics[width=0.49\textwidth]{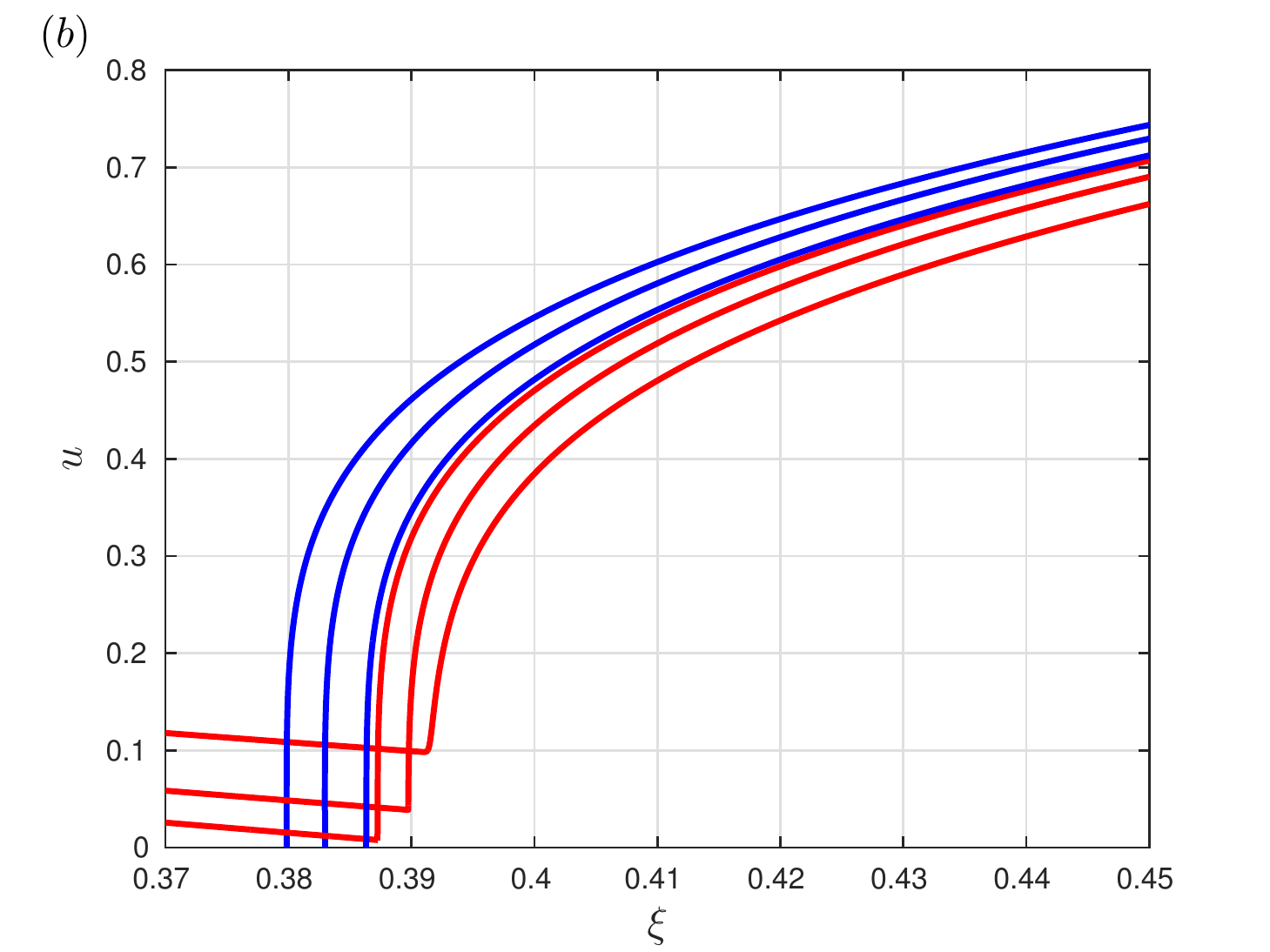}
\caption{Panel (a) shows some representative trajectories emanating from $x_0^* \approx 0.338$ for $m=2$.
More precisely the red and blue trajectories begin at $x_0^* \pm \hat{\delta}$ for
$\hat{\delta} = 10^{-2}, 10^{-4}, 10^{-6}, \, 10^{-8}, \, 10^{-11}$ and $10^{-14}$.
The black line shows the artificially extrapolated linear behaviour. Panel (b)
shows some representative trajectories emanating from $x_0^* \approx 1.165$ for $m=4$.
In this case, the trajectories begin at $x_0^* \pm \hat{\delta}$ for $\hat{\delta} = 10^{-1}, \, 10^{-2}$
and $10^{-3}$. Notably, in the latter case, despite only resolving $x_0^*$ to 3 significant digits,
a good estimate of $A_-$ has already been obtained.}
\label{example-shooting}
\end{figure}
% Code: /home/jmfoster/Documents/Shooting_for_reversing_solutions/Pelinovsky_shooting/automated_shooting_scheme/automated/auto_track_semi_automatic/binary_search/binary_search_lower_branch/ex_shoot_a.m
% Code: /home/jmfoster/Documents/Shooting_for_reversing_solutions/Pelinovsky_shooting/automated_shooting_scheme/automated/auto_track_semi_automatic/binary_search/binary_search_lower_branch/ex_shoot_b.m

\begin{figure}          \centering
\includegraphics[width=\textwidth]{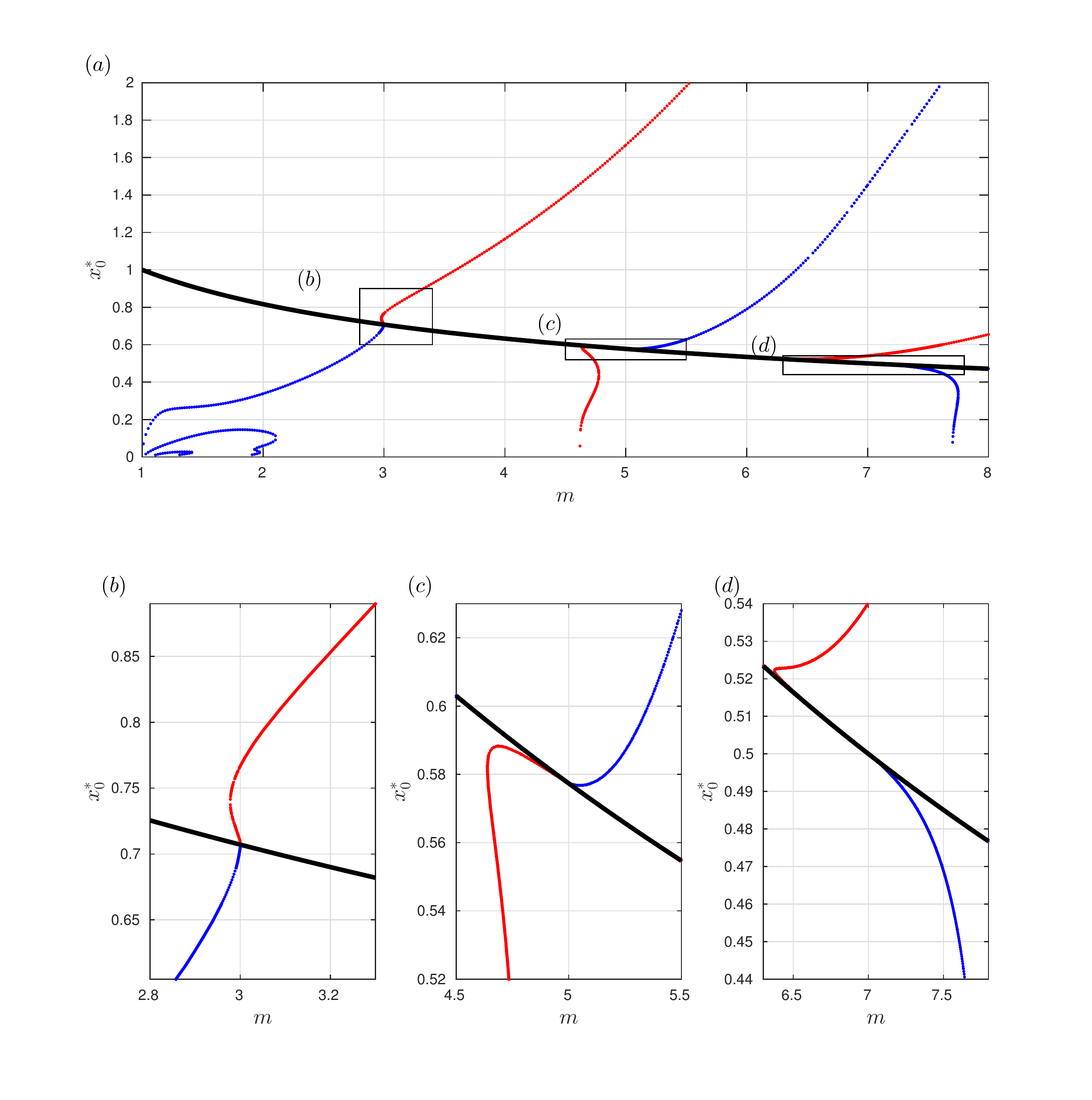}\\
\caption{The variation of $x_0^*$ versus $m$. The red, blue and black curves indicates values of $x_0^*$ that define trajectories terminating at the near-field equilibrium point with $A_- > 0$, $A_- < 0$, and $A_- = 0$ respectively. Panels (b)-(d) show zoomed-in regions from panel (a) near $m = 3$, $m = 5$, and $m = 7$ respectively.}
\label{all-sols-negative-t}
\end{figure}
%Code: /home/jmfoster/Documents/Shooting_for_reversing_solutions/Pelinovsky_shooting/THE_DADDIO/plotting_beast_v3.m
% Note: Need to `flick' figure to one side of screen to get annotations into right places

\begin{figure}          \centering
\includegraphics[width=0.49\textwidth]{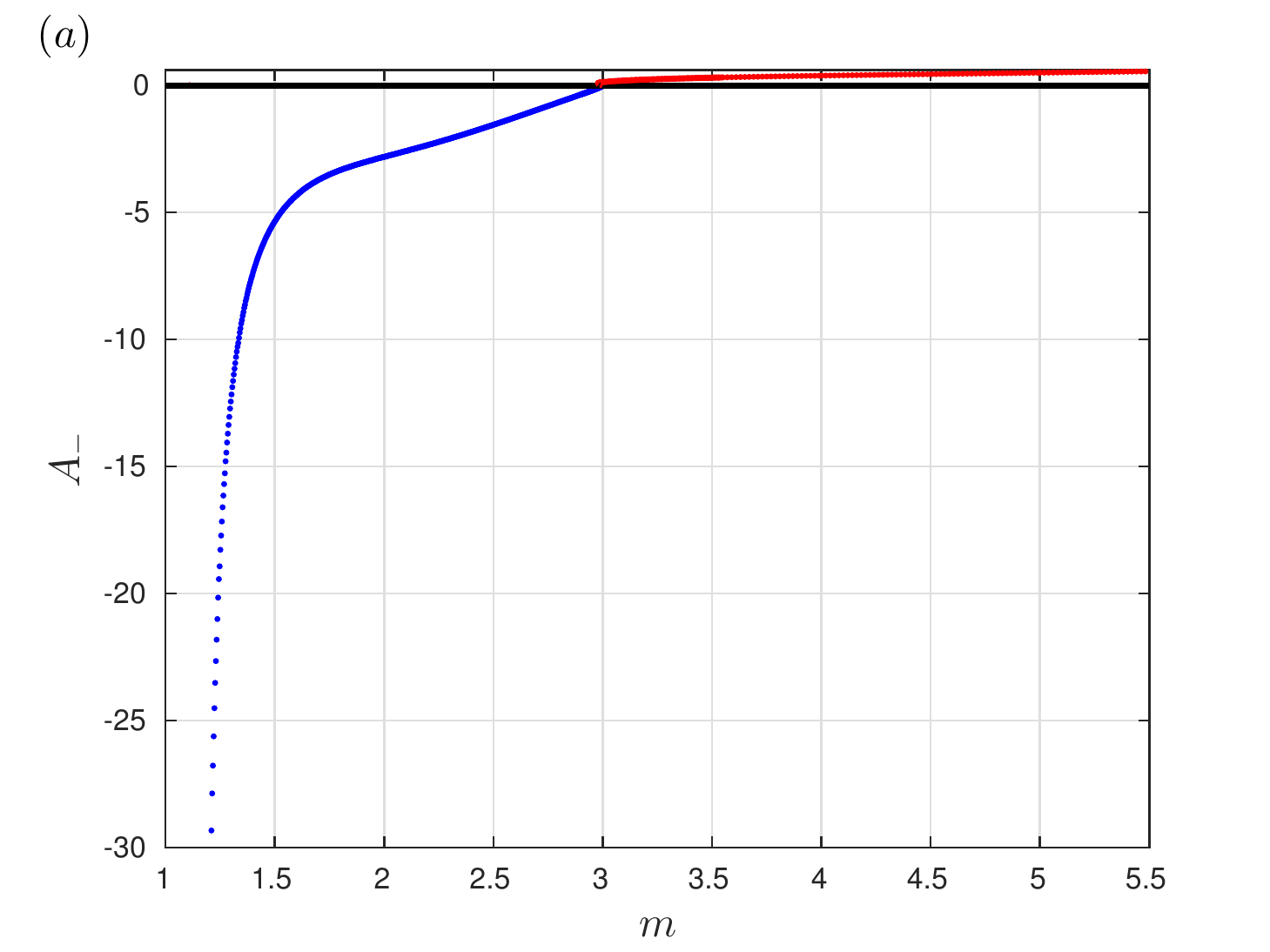}
\includegraphics[width=0.49\textwidth]{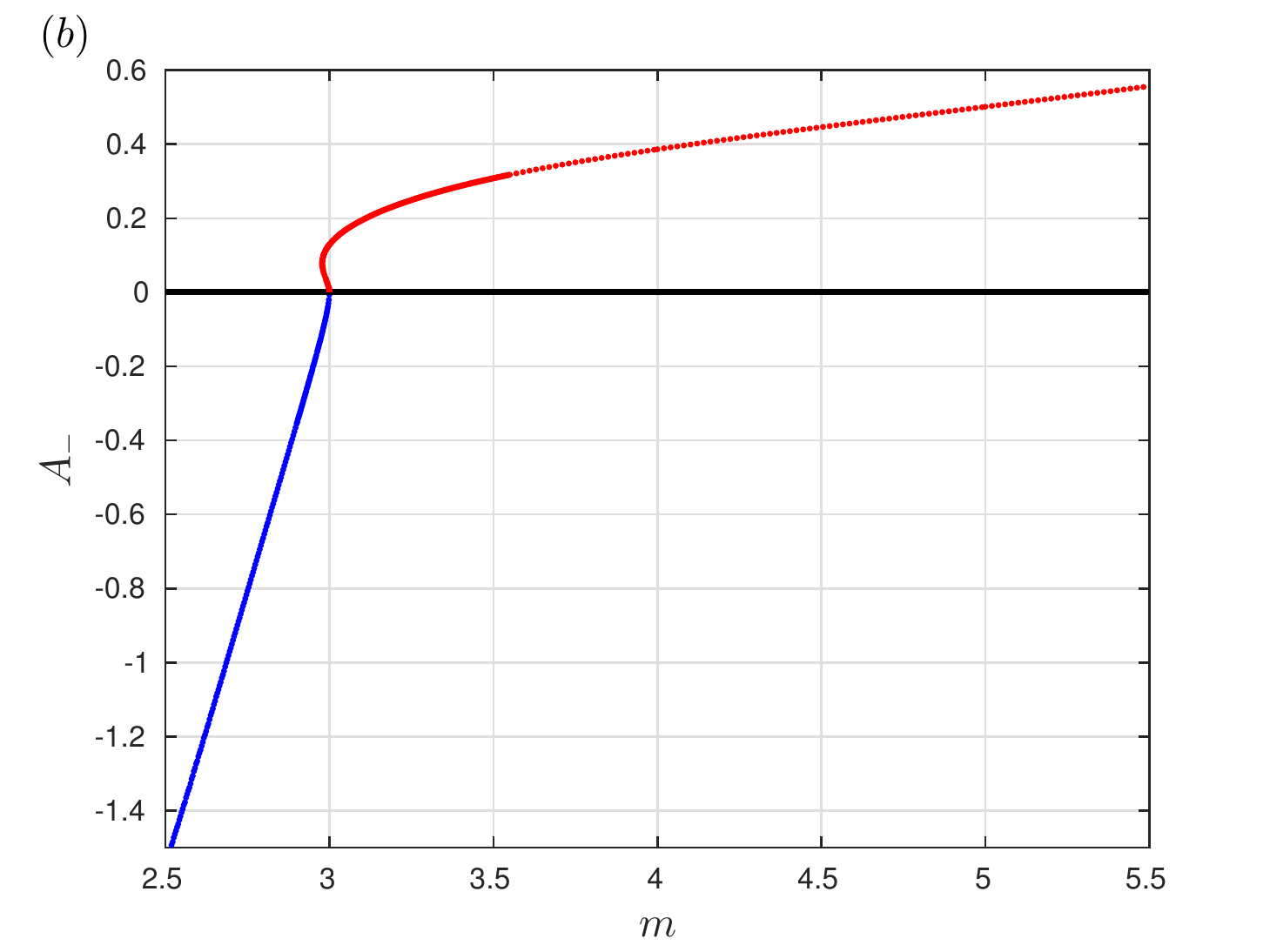}
\caption{Panel (a) shows the variation of $A_-$ along the red and blue curves emanating from the black curve near $m=3$. Panel (b) shows the same plot zoomed in on positive values of $A_-$.}
\label{negative-t-A-minus}
\end{figure}
%Code: /home/jmfoster/Documents/Shooting_for_reversing_solutions/Pelinovsky_shooting/THE_DADDIO/plotting_beast_A_plus.m
%\includegraphics[width=0.49\textwidth]{new_tile_5.eps}

\begin{figure}[h!]        \centering
\includegraphics[width=0.49\textwidth]{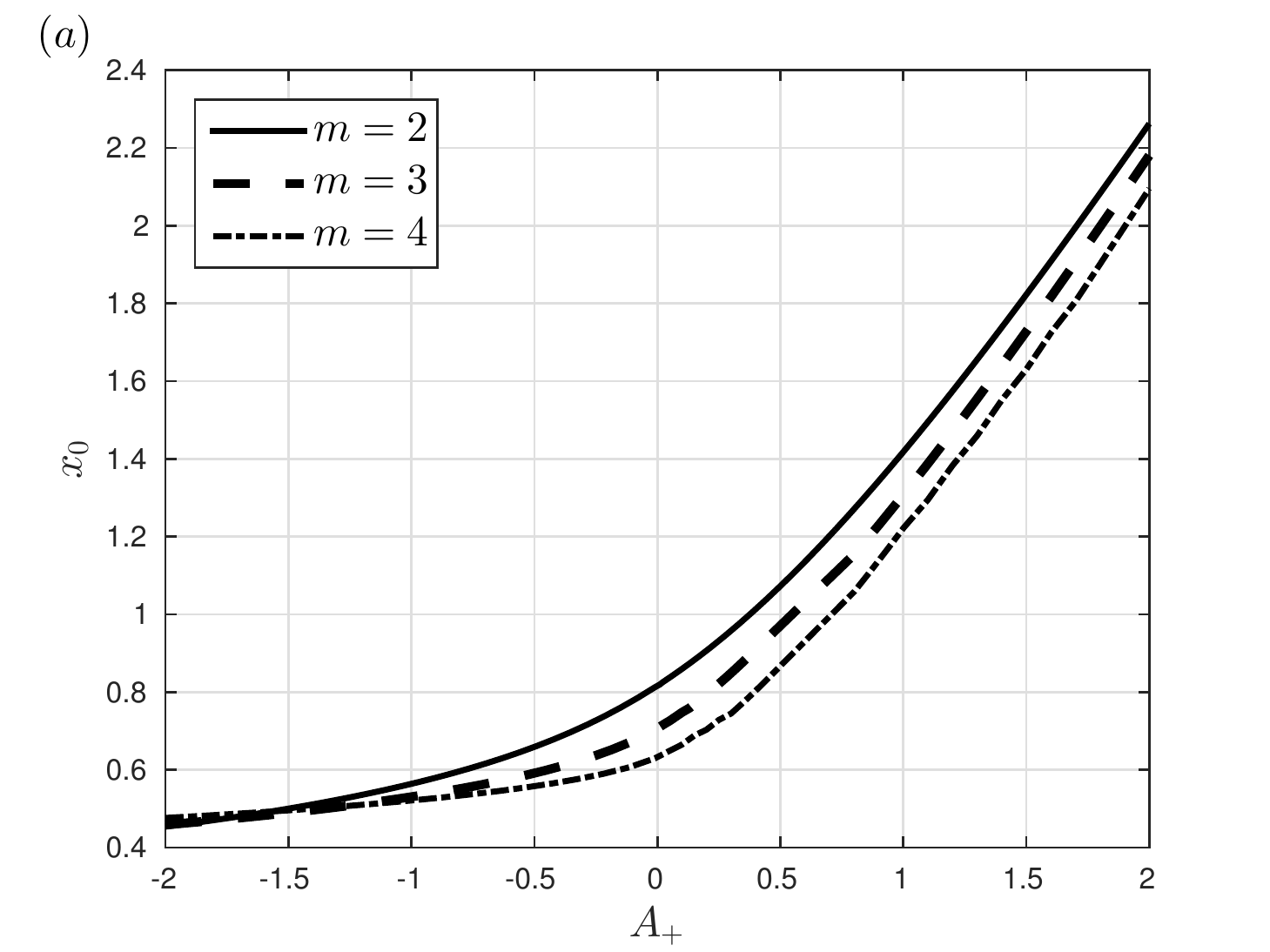}
\includegraphics[width=0.49\textwidth]{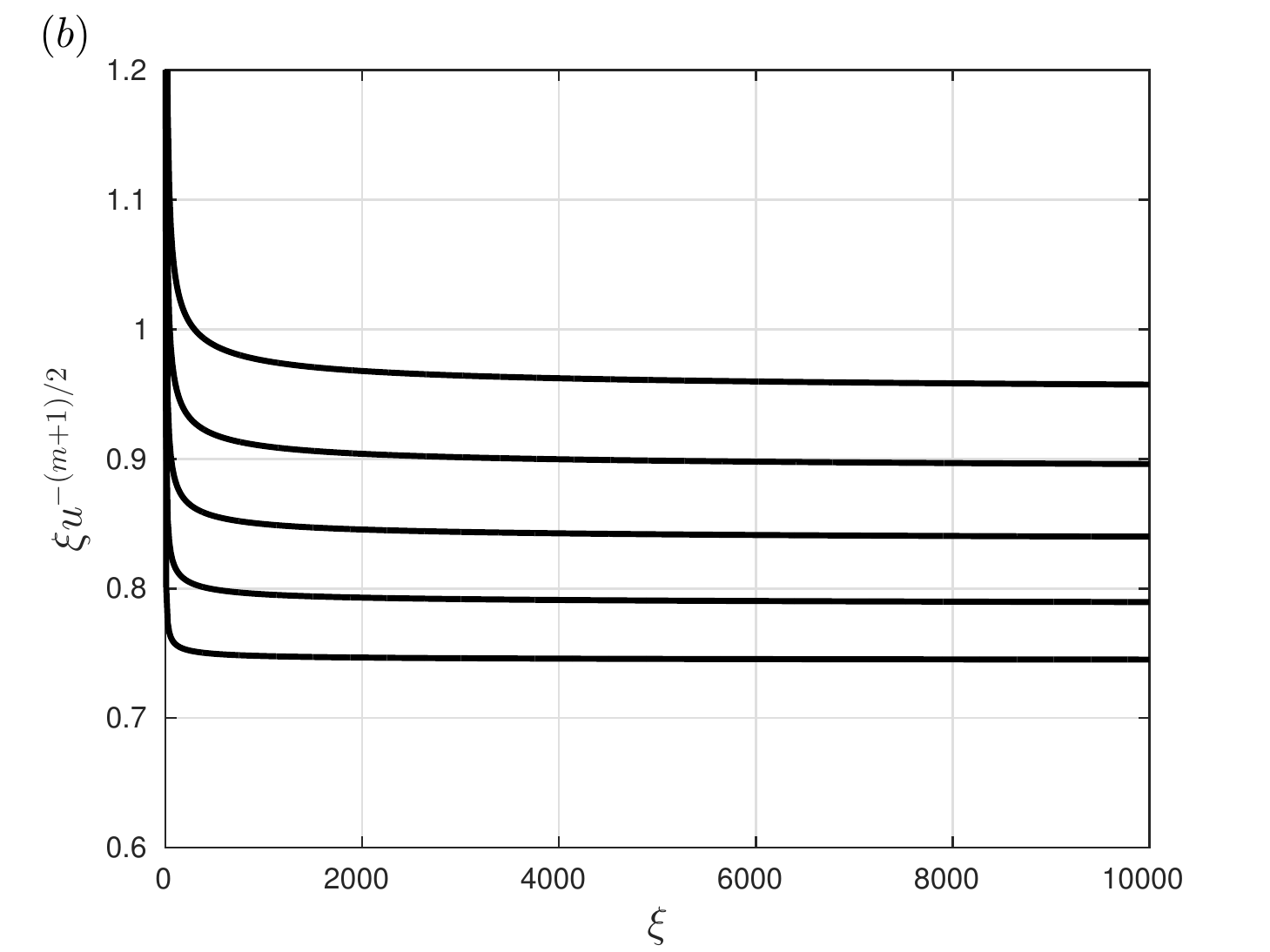}
\caption{Panel (a):  Plots of the variation of $x_0$ with $A_+$ for various different values of $m=2,\, 3$ and $4$.
Panel (b): Plots of the trajectories emanating from $A_+ = 0.1,\, 0.2, \, 0.3, \, 0.4$ and $0.5$ for $m=3$.
The constant to which these trajectories tend in the far-field is selected to be the corresponding value of $x_0$. }
\label{t-positive-map}
\end{figure}
% Code: /home/jmfoster/Documents/Shooting_for_reversing_solutions/Pelinovsky_shooting/t_positive_codes/outward.m
% Code: /home/jmfoster/Documents/Shooting_for_reversing_solutions/Pelinovsky_shooting/t_positive_codes/plotting.m

\subsection{\label{t-pos-solutions}Solutions for $H_+$ ($t>0$)}

Having successfully found suitable solutions for $H_-$, we now proceed to
compute suitable solutions for $H_+$. As discussed in \S\ref{H-plus-connect},
we can numerically construct a unique trajectory from small to infinite values of $H_+$.
To do so, we integrate the system (\ref{zero-dynamics}) from the equilibrium point $(A_+,0,0)$
in the near-field towards the equilibrium point $(x_0,0,0)$ of the system (\ref{infinity-dynamics})
in the far-field.  The numerical procedure is carried out as follows:

\begin{remunerate}
\item Select a value of $A_+ \in {\mathbb R} \backslash \{0\}$. Since $(A_+,0,0)$ is an equilibrium point,
it is not possible to escape $(A_+,0,0)$ in finite time. We therefore begin integration of the system (\ref{zero-dynamics})
by taking a small step, say $\epsilon$, away from $(A_+,0,0)$ using the relevant asymptotic behaviour.
Using (\ref{center-manifold-zero}) and (\ref{center-manifold-zero-dynamics}) for $A_+ > 0$, we find that
a trajectory exiting the equilibrium point along the center manifold, $W_c(A_+,0,0)$, has the local asymptotic behaviour
\begin{equation}
\label{numeric-begin-t-positive-positive}
\left\{ \begin{array}{l}
\xi = A_+ + \epsilon, \\
u = \left( \frac{m+1}{2} A_+ \right)^{-1} \epsilon + \mathcal{O}(\epsilon^2), \\
w = \left( \frac{m+1}{2} A_+ \right)^{-(m+1)} \epsilon^m + \mathcal{O}(\epsilon^{\min\{m+1,2m-1\}}),\end{array} \right.
\quad \mbox{for} \quad A_+>0,
\end{equation}
for a small positive value of $\epsilon$. By contrast, using (\ref{stable-manifold-zero}) and (\ref{stable-manifold-zero-dynamics}) for $A_+<0$,
we find that a trajectory along the unstable manifold $W_u(A_+,0,0)$ has the local asymptotic behaviour
\begin{equation}
\label{numeric-begin-t-positive-negative}
\left\{ \begin{array}{l}
\xi = -|A_+| + \epsilon, \\
u = \left( \frac{m+1}{2} |A_+| m \right)^{\frac{1}{m}} \epsilon^{\frac{1}{m}} + \mathcal{O}(\epsilon), \\
w = \frac{1}{m} \left( \frac{m+1}{2} |A_+| m \right)^{\frac{m+1}{m}} \epsilon^{\frac{1}{m}} + \mathcal{O}(\epsilon),\end{array} \right.
\quad \mbox{for} \quad A_+<0,
\end{equation}
for a small positive value of $\epsilon$.
Having selected values for both $A_+$ and $\epsilon$, either (\ref{numeric-begin-t-positive-positive})
or (\ref{numeric-begin-t-positive-negative}) define unique (pseudo-)initial conditions to begin integrating
the system (\ref{zero-dynamics}) in the direction of increasing time $\tau$ towards the far-field.

\item We proved in Corollary \ref{coral-1}, that the ultimate fate of all such trajectories,
in variables $(x,y,z)$, is approaching the equilibrium state $(x_0,0,0)$ for some $x_0 \in [0,\infty)$.
Thus, by continuing integration of the system (\ref{zero-dynamics}) to some large value of $\tau$,
denoted by say $\tau_\infty$, and reading off the value $\xi u^{-(m+1)/2} \approx x_0$ at $\tau = \tau_\infty$,
we can obtain an arbitrarily accurate approximation of the corresponding value of $x_0$ that is obtained
in the far-field --- a higher degree of accuracy can be achieved by simply increasing the value of
$\tau_\infty$. For this purpose we found {\tt ode45} with the majority of the default setting
to be sufficiently robust. To ensure high numerical accuracy, at the cost of a relatively small
increase in computation time, both {\tt AbsTol} and {\tt RelTol} were decreased to $10^{-10}$.
In contrast to the case for solutions $H_-$, we found it unnecessary to `switch' from integrating
the near-field system (\ref{zero-dynamics}) to the far-field system (\ref{infinity-dynamics}).
Typically, we found that taking  $\tau_\infty \in (10^4,10^5)$ gave an approximation of $x_0$ correct to 8 significant digits.
\end{remunerate}

Carrying out this procedure for a variety of choices of $A_+$ we are able to trace out
the form of the piecewise $C^1$ map between $A_+$ and $x_0$ defined earlier in (\ref{two-maps}).
In figure \ref{t-positive-map}, we show this map for $m=2, \, 3$ and $4$ (see panel (a)), as well as
some representative trajectories of the system (\ref{zero-dynamics}) for $m=3$ (see panel (b)).
In addition to the results shown, other computations for different values of $m$ were also
carried out and it appears generic that $x_0$ is a monotonically increasing function of $A_+$.
Crucially, it appears that range of the map (\ref{two-maps}) is the entire semi-axis $\mathbb{R}^+$ for $x_0$.

\subsection{\label{num-sum}Summary of numerical results}

We have demonstrated that: (i) for each value of $m>1$ there exists
\emph{at least} one value of $x_0=x_0^*$ (different from the trivial case $x_0=x_Q$)
that defines a trajectory emanating from $(x_0^*,0,0)$ and terminating at a $(A_-,0,0)$,
and thus a suitable solution for $H_-$, and; (ii) for every value of $A_+ \in {\mathbb R}$
there exists a unique corresponding value of $x_0$, thereby defining an infinite family of
suitable solutions for $H_+$. The one remaining step is therefore to invoke the matching
condition (\ref{far-field-matching}). This condition is equivalent to requiring that the
far-field behaviour of $H_+$ is characterized by $x_0=x_0^*$. Thus, given a solution for
$H_-$, the matching condition (\ref{far-field-matching}) specifies a unique choice of
$x_0=x_0^*$, a unique $A_+$, and thus a unique solution for $H_+$, thereby closing the problem.

The solutions found here for $m=3$ and $4$ show a qualitative,
although not quantitative, agreement with those reported in \cite{Foster}. Here, we found that
\begin{equation} \label{m-is-3-vals}
m = 3 : \quad A_- \approx 0.129, \quad A_+ \approx 0.154 \quad \mbox{and} \quad x_0^* \approx 0.767
\end{equation}
and
\begin{equation} \label{m-is-4-vals}
m = 4 : \quad A_- \approx 0.386, \quad A_+ \approx 0.794 \quad \mbox{and} \quad x_0^* \approx 1.165.
\end{equation}
In Foster \emph{et. al.} \cite{Foster}, they claimed that for $m=3$: $A_- \approx 0.144$,
$A_+ \approx 0.0958$, and $x_0^* \approx 0.765$, whereas for $m = 4$: $A_- \approx 0.386$,
$A_+ \approx 0.341$, and $x_0^* \approx 0.980$. Additionally, they claim another suitable solution for $m=2$: $A_- \approx 0.00135$, $A_+ \approx 0.0102$, and $x_0^* \approx 0.817$. In contrast, here we found that no such solution with $A_->0$ exists. Notably this value of $x_0^*$ reported in \cite{Foster} for $m = 2$ is very close
to the value of $x_Q = \sqrt{2/(m+1)}$. For $m=2$, using our numerical approach we have been able to identify three
other solutions with $A_-<0$, namely:
\begin{equation} \label{m-is-2-vals}
m=2: \quad \left\{ \begin{array}{l}
 A_- \approx -2.804, \quad A_+ \approx -4.322, \quad x_0^* \approx 0.338,  \\*[1mm]
 A_- \approx -0.932  , \quad  A_+ \approx -30.625,  \quad x_0^* \approx 0.137, \\*[1mm]
 A_- \approx -0.546, \quad A_+ \approx -166.623, \quad x_0^* \approx 0.0592. \end{array}  \right.
\end{equation}
We believe that the origin of these discrepancies is due to the low accuracy of the numerical scheme
used in \cite{Foster}. Indeed, in \cite{Foster}, the solutions for $H_-$ were computed by identifying
the value of $A_-$ which characterizes solutions in the near-field that extend into far-field with
the requisite behaviour, as in panel (a) of figure \ref{return-to-zero}. Solutions for $H_+$ were
computed by finding the value of $A_+$ inferred (via shooting from the far-field toward the near-field)
by invoking the matching condition in the far-field as in panel (b) on figure \ref{t-positive-map}.
Both numerical methods used in \cite{Foster} are ill-posed.
Here, we pose the numerical problem as a shooting scheme for uniquely defined
piecewise $C^1$ scalar functions, \emph{i.e.} the maps defined in (\ref{two-maps}) and (\ref{maps-negative}).
We therefore believe that the results obtained here are more reliable than those in \cite{Foster}.

\section{\label{conc}Conclusion}

\begin{figure}[h!]        \centering
\includegraphics[width=0.49\textwidth]{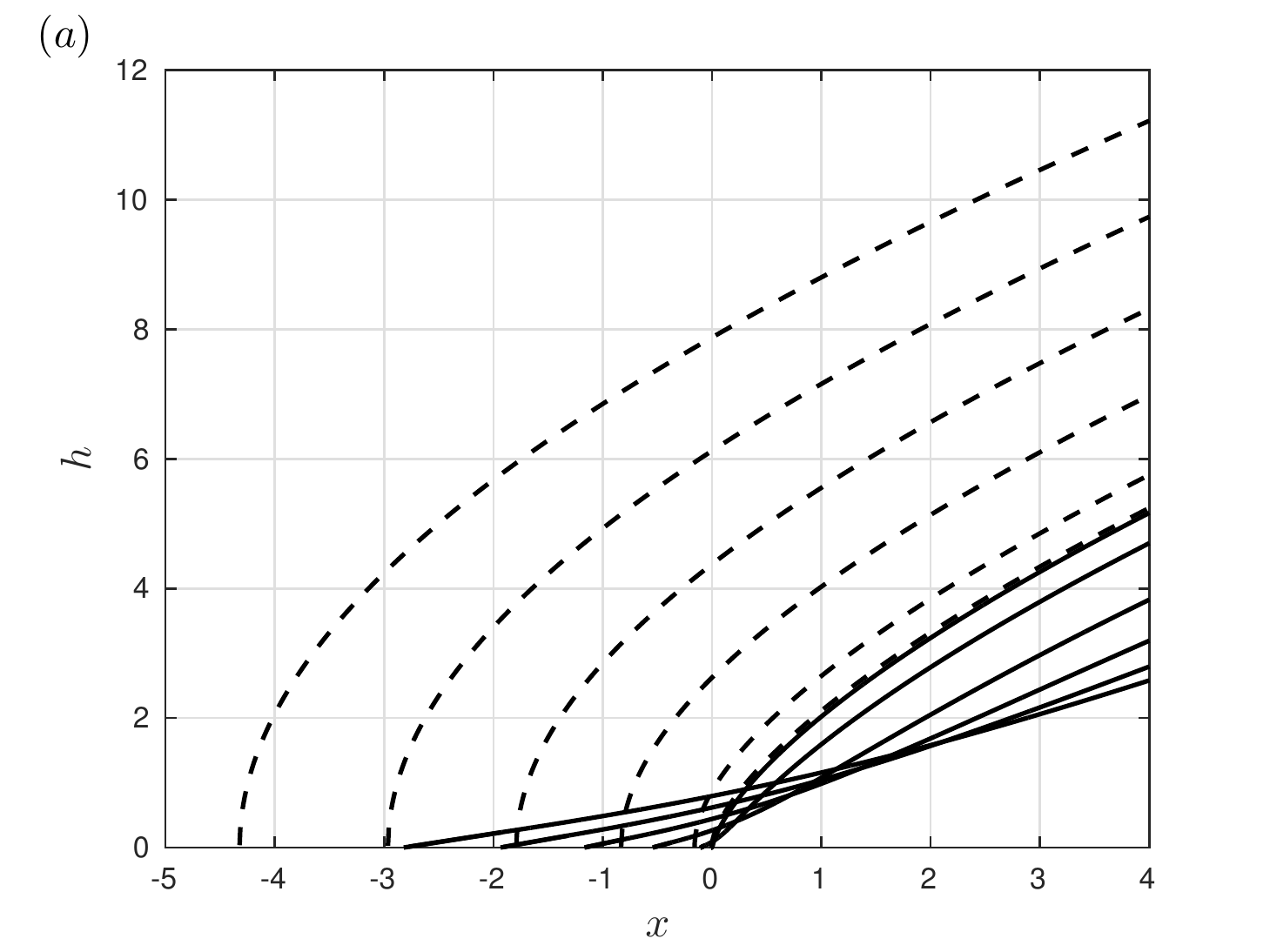}
\includegraphics[width=0.49\textwidth]{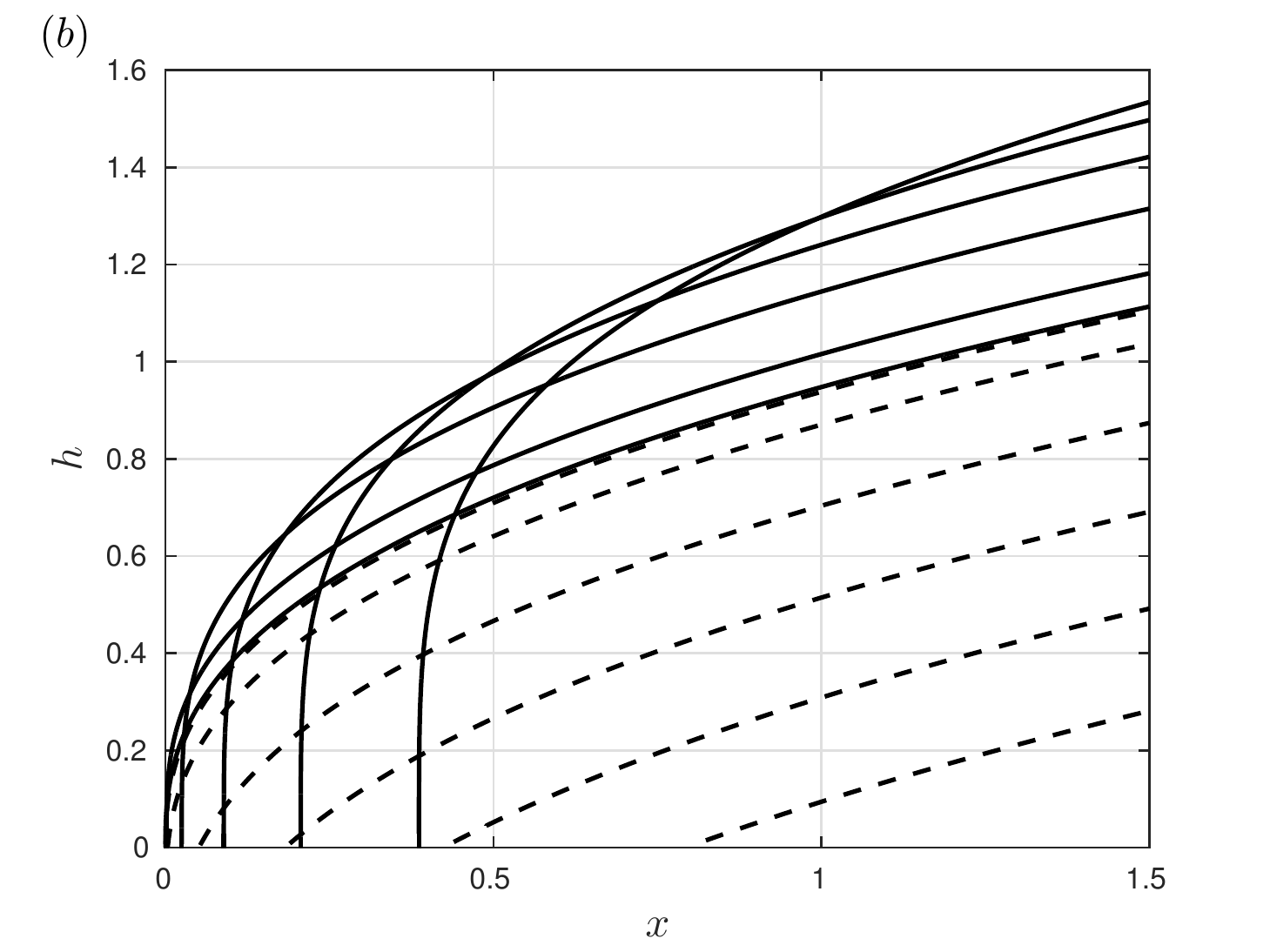}
\caption{Representative plots of $h(x,t)$ at 10 equally spaced values of $t$ between $-1$ and $1$, 
and for $t = \pm 10^{-3}$ to demonstrate the continuity of $h(x,t)$ across $t=0$. Solid and dashed 
curves show the solution for $t<0$ and $t>0$ respectively. Panel (a) shows the anti-reversing dynamics 
for $m=2$, $A_- \approx -2.804$, $A_+ \approx -4.322$ and $x_0^* \approx 0.338$. Panel (b) shows the reversing 
dynamics for $m=4$, $A_- \approx 0.386$, $A_+ \approx 0.794$ and $x_0^* \approx 1.165$.}
\label{generic-figure-name}
\end{figure}

This work has focused on constructing local (in both space and time) self-similar reversing
and anti-reversing solutions to the nonlinear diffusion equation (\ref{heat}) with $m>1$ and $n=0$. We have
demonstrated how the dynamical theory combined with the numerical scheme can be used to furnish suitable solutions to the
differential equations (\ref{ode}) for $H_-$ and $H_+$. Via the self-similar reductions (\ref{reduction}),
the solutions to these differential equations can be transformed into physically meaningful solutions
for $h(x,t)$ to the nonlinear diffusion equation (\ref{heat}) with $m>1$ and $n=0$, which is completed 
with the no flux boundary condition (\ref{zero-flux2}) and the condition the $h|_{x=\ell(t)}=0$. 
In this final section we shall discuss the connection
between the self-similar solutions found here and the dynamics of the model (\ref{heat}).

It is well-known, and can be readily verified, that the nonlinear diffusion equation with absorption (\ref{heat})
admits two different travelling wave solutions. Seeking such solutions near a left interface
$x=\ell(t)$ that have the form $h(x,t) \sim A(t) (x-\ell(t))^\alpha$, for some function $A(t)$ and constant $\alpha$, we find that
\begin{eqnarray}
\label{non-lin-nose} h \sim (-m \dot{\ell})^{1/m} (x-\ell(t))^{1/m}  \quad \mbox{as} \quad x \searrow \ell(t),
\quad \mbox{for} \quad \dot{\ell} < 0,\\*[3mm]
\label{lin-nose} h \sim (\dot{\ell})^{-1} (x-\ell(t)) \quad \mbox{as} \quad x \searrow \ell(t),
\quad \mbox{for} \quad \dot{\ell} > 0.
\end{eqnarray}
The former, is an advancing wave local to a left interface whose motion is driven by
diffusion, whereas the latter is a receding wave driven by absorption. This study has
therefore elucidated the process by which the wave (\ref{non-lin-nose}) \emph{becomes} (\ref{lin-nose}), giving rise to a reversing interface, or vice versa, giving rise to an anti-reversing interface.

In \S\ref{t-neg-shooting} we used the result of Lemma \ref{lemma-continuation}
to numerically construct suitable solutions for $H_-$. For each value of $m>1$ we
identified \emph{at least} one suitable solution, defined by a pair of values of $A_-$ and $x_0^*$,
in addition to the exact solution (\ref{exact-solution}) --- in the original time and space variables,
this exact solution corresponds to a steady solution for $h(x,t)$ and thus does not constitute a
reversing nor an anti-reversing solution. In \S\ref{t-pos-solutions}, we used Lemma
\ref{lemma-connection} to formulate a numerical scheme for constructing solutions for $H_+$
defined by pairs of values of $A_+$ and $x_0^*$. We showed that the map (\ref{t-positive-map})
is one-to-one and its range is the entire semi-axis $\mathbb{R}^+$ for $x_0$. Importantly, for
each value of $m>1$, we found that: (i) if $A_+<0$ then $x_0<x_Q$, and (ii) if $A_+>0$ then
$x_0>x_Q$ where $x_Q=\sqrt{2/(m+1)}$. The final stage in constructing solutions for $h(x,t)$
is to invoke the matching condition (\ref{far-field-matching}) that ensures continuity of
$h(x,t)$ across $t=0$. Owing to the aforementioned properties of (\ref{t-positive-map})
we are forced to reject any solution for $H_+$ that is defined by a trajectory with $A_-<0$
and $x_0^*>x_Q$ or $A_->0$ and $x_0^*<x_Q$ on the basis that it necessarily cannot match to
solution for $H_-$. In summary we have found that for $1<m<3$ up to 5 different solutions
are available with $A_-,A_+<0$. For $m \gtrsim 2.97$ there is at least one solution with
$A_-,A_+>0$. For $7<m<7.75$ an additional branch of solutions with $A_-,A_+<0$ emerges, 
whereas for $m \gtrsim 6.42$, there is another branch of solutions with $A_-,A_+ > 0$.
For $m>8$ it seems quite possible that yet more branches of solutions will emerge.

There is a distinct difference between the interpretation of solutions with $A_-,A_+>0$
in terms of the original model (\ref{heat}) compared to those with $A_-,A_+<0$.
The former, correspond to a reversing solution where the left interface
advances for $t<0$, with the behaviour (\ref{non-lin-nose}), and then subsequently
recedes for $t>0$ with the behaviour (\ref{lin-nose}). Contrastingly, the latter
corresponds to an anti-reversing solution where the interface recedes for $t<0$,
with the form (\ref{lin-nose}), and then advances according to the form (\ref{non-lin-nose}) for $t>0$. 
One representative local solution for $h(x,t)$ for both types of behaviour --- 
one reversing and one anti-reversing --- are shown in figure \ref{generic-figure-name}.

Some natural open questions raised by this study are: (i) whether any self-similar solutions with non-monotone profiles (in $\xi$) exist --- \emph{i.e.} solutions that do not satisfy (\ref{(ii)}); (ii) whether the self-similar
solutions identified here are stable in the context of the model (\ref{heat}), and; 
(iii) if more than one reversing or anti-reversing solution is stable for
a particular value of $m$, what is the mechanism for selecting the appropriate
self-similar solution at a particular reversing or anti-reversing event.

\vspace{0.25cm}\emph{}

\noindent{\bf Acknowledgements.}
D.P. thanks M. Chugunova and R. Taranets, while J.F. thanks J. R. King and A. D. Fitt for
useful discussions regarding this project. J.F. is supported by a postdoctoral fellowship at McMaster University.
He thanks B. Protas for hospitality and many useful discussions. 
A part of this work was completed during the visit of D.P. to Claremont Graduate University.
The work of D.P. is supported by the Ministry of Education
and Science of Russian Federation (the base part of the state task No. 2014/133, project No. 2839).

\end{document}